\documentclass[10pt,a4paper]{article}

\usepackage{amssymb}
\usepackage{latexsym}
\usepackage{amsmath}

\usepackage{graphicx}

\usepackage{stmaryrd}

\newcommand{\MBB}{\mathbb}

\newcommand{\MCAL}{\mathcal}

\newcommand{\PAR}{\bindnasrepma}
\newcommand{\TENS}{\otimes}

\newcommand{\DS}{\displaystyle}

\newtheorem{lemma}{Lemma}[section]
\newtheorem{theorem}{Theorem}[section]
\newtheorem{definition}{Definition}[section]

\newtheorem{proposition}{Proposition}[section]
\newtheorem{example}{Example}[section]
\newtheorem{remark}{Remark}

\newenvironment{proof}{\begin{flushleft}{\bf Proof:} \ \ }{\end{flushleft}}

\begin{document}
\title{A New Linear Time Correctness Condition for Proof Nets of Multiplicative Linear Logic}
\author{Satoshi Matsuoka \\
  National Institute of Advanced Industrial \\
  Science and Technology (AIST), \\
        1-1-1 Umezono, Tsukuba, Ibaraki, 305-8565 Japan\\
        {\tt matsuoka@ni.aist.go.jp}}
\date{}
\maketitle
\begin{abstract}
  In this paper, we give a new linear time correctness condition for proof nets of Multiplicative Linear Logic
  without units.
  Our approach 
  is based on a rewriting system over trees.
  We have only three rewrite rules.
  Compared with previous linear time correctness conditions, our system is surprisingly simple, intuitively appealing, and easy to implement. 
\end{abstract}
\section{Introduction}
More than three decades ago, J.Y. Girard introduced the notion of proof nets for unit free Multiplicative Linear Logic (for short, MLL)\cite{Gir87}.
It is a parallel syntax for MLL proofs, removing redundancy of sequent calculus proofs.
In \cite{Gir87}, he introduced MLL proof structures, which are graphs whose nodes are labeled by MLL formulas
and then defined MLL proof nets as sequentializable MLL proof structures,
where an MLL proof structure is sequentializable if
one can recover a sequent calculus proof from it by a decomposition procedure. 
Moreover he introduced a topological property called the {\it long trip condition} for MLL proof structures
and showed that an MLL proof structure is an MLL proof net if and only if it satisfies the long trip condition.
Such a characterization is called a {\it correctness condition} for MLL proof nets.
Since then many other correctness conditions have been given for MLL and its variants or extensions by many researchers.

Complexity questions about correctness conditions arise naturally.
The first linear time correctness condition for MLL is given in \cite{Gue99},
which is based on contractability condition \cite{Dan90}.
Other linear time correctness conditions are given in \cite{MO00},
which are based on that of essential nets, which are an intuitionistic variant of MLL proof nets.
Moreover de Naurois and Mogbil introduce a correctness condition for MLL and their extensions
based on topological conditions of arbitrarily selected
one DR-graph (\cite{DR89}) 
and showed that they are NL-complete \cite{deNM11}.

In this paper we introduce a new linear time correctness condition for MLL.
It is based on that of \cite{deNM11}.
Although de Naurois and Mogbil showed their correctness condition is NL-complete, 
its linear time termination cannot be derived from their presentation in \cite{deNM11} directly.
In order to establish the linear time correctness condition,
we define a {\it rewriting system} over {\it deNM-trees}, where
a deNM-tree is a labeled-tree, whose definition is inspired by the correctness condition in \cite{deNM11}.
The rewriting system has only three rewrite rules, which is remarkably simple. 
In the rewriting system, an active node flows in a deNM-tree, reducing nodes by the rewriting rules.
However, the rewriting system may lead to quadratic time termination in the worse case.
In order to fix the situation, we introduce more sophisticated data structures and a rewriting strategy.
Thanks to them, we can achieve the linear time termination.

Compared with \cite{Gue99} and \cite{MO00}, our correctness condition is
surprisingly simple and intuitively appealing.
While the correctness condition in \cite{Gue99} has to use a non-local jump rule called the new rule,
all three rewriting rules in our system are strictly local. 
In addition, any of correctness conditions in \cite{MO00} is rather complex,
since they need complicated queries about directed paths or a synchronization mechanism.
Our rewriting system consists of just three rewrite rules. 

Besides, we also gave an implementation \cite{Mat19a} for our linear time correctness condition.
Compared with a naive quadratic implementation for the correctness condition in \cite{deNM11},
our new implementation is much faster, especially in bigger MLL proof structures. 
As far as we know, there are no publicly available implementations for linear time correctness conditions in \cite{Gue99} and \cite{MO00}.

\section{Multiplicative Linear Logic, Proof Structures and Proof Nets}
\subsection{Multiplicative Linear Logic}
We introduce the system of Multiplicative Linear Logic (for short MLL).
We define {\it MLL formulas}, which are denoted by $F, G, H, \ldots$,  by the following grammar:
\[
F ::= p \, \, | \, \, p^\bot \, \, | \, \, F \TENS G \, \, | \, \, F \PAR G
\]
The negation of $F$, which is denoted by $F^\bot$ is defined as follows:
\[
\begin{array}{lcl}
  {(p)}^\bot & = & p^\bot \\
  {(p^\bot)}^\bot & = & p \\
{(F \TENS G)}^\bot & = & G^\bot \PAR F^\bot \\
{(F \PAR G)}^\bot  & = & G^\bot \TENS F^\bot 
\end{array}
\]
The formula $p$ is called an {\it atomic} formula.
In this paper, we only consider the logical system with only one atomic formula:
We can reduce the correctness condition with many atomic formulas to this simplified case by forgetting the information. 
We denote {\it multisets of MLL formulas} by $\Lambda, \Lambda_1, \Lambda_2, \ldots$.
An MLL sequent is a multiset of MLL formulas $\Lambda$.
We write an MLL sequent $\Lambda$ as $\vdash \Lambda$. 
The inference rules of MLL are as follows:
\[
\begin{array}{llcll}
{\rm ID}    & \frac{}{\DS \vdash p^\bot, p} & & & \\
& & & & \\
\TENS & \frac{\DS \vdash \Lambda_1,  F \quad \quad \vdash \Lambda_2, G}{\DS \vdash \Lambda_1, \Lambda_2, F \TENS G}
& &
\PAR & \frac{\DS \vdash \Lambda, F, G}{\DS \vdash \Lambda, F \PAR G}
\end{array}
\]
We note that we restrict the ID-axiom to that with only atomic formula $p$ and its negation $p^{\bot}$. 
We omit the {\it cut} rule that has the form
\[
{\rm Cut} \, \, \frac{\DS \vdash \Lambda_1,  F \quad \quad \vdash \Lambda_2, F^{\bot}}{\DS \vdash \Lambda_1, \Lambda_2}
\]
because it can be identified with the $\TENS$-rule for our purpose.
\subsection{MLL Proof Nets}
Next we introduce MLL proof nets.
Figure~\ref{figMLLLinks} shows the {\it MLL links} we use.
Each MLL link has a few MLL formulas.
Such an MLL formula is a conclusion or a premise of the MLL link, which is specified as follows:
\begin{enumerate}
\item In an ID-link, each of $p$ and $p^\bot$ is called a conclusion of the link.
\item In a $\TENS$-link, each of $F$ and $G$ is called a premise of the link and
  $F \TENS G$ is called a conclusion of the link.
\item In a $\PAR$-link, each of $F$ and $G$ is called a premise of the link and
  $F \PAR G$ is called a conclusion of the link.
\end{enumerate}
In the definition above $F$ is called left premise and $G$ right premise. 
\begin{figure}[htbp]
\begin{center}
  \includegraphics[scale=0.6]{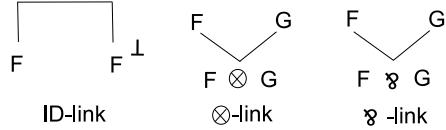}
\end{center}
 \caption{MLL Links}
 \label{figMLLLinks}
\end{figure}
An MLL {\it proof structure} $\Theta$ is a set of MLL links that satisfies the following conditions:
\begin{enumerate}
\item For each link $L$ in $\Theta$,
  each conclusion of $L$ is a premise of at most one link other than $L$ in $\Theta$.
\item For each link $L$ in $\Theta$,
  each premise of $L$ must be a conclusion of exactly one link other than $L$ in $\Theta$.
\end{enumerate}
A formula occurrence $F$ in an MLL proof structure $\Theta$ is a conclusion of $\Theta$ if
$F$ is not a premise of any link in $\Theta$. 

An MLL {\it proof net} is an MLL proof structure that is constructed by the rules in Figure~\ref{figMLLProofNets}.
Note that each rule in Figure~\ref{figMLLProofNets} has the corresponding inference rule in the MLL sequent calculus. 
All MLL proof structures are not necessarily an MLL proof net.
\begin{figure}[htbp]
\begin{center}
  \includegraphics[scale=0.6]{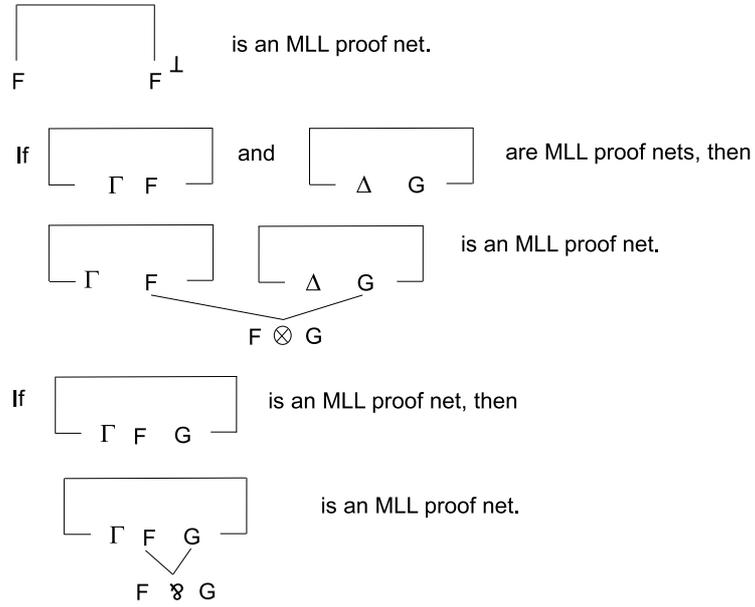}
\end{center}
 \caption{Definition of MLL Proof Nets}
 \label{figMLLProofNets}
\end{figure}
Next we introduce a characterization of MLL proof nets using the notion of DR-switchings.
A DR-switching $S$ for an MLL proof structure $\Theta$ is a
function from the set of $\PAR$-links in $\Theta$ to $\{ 0, 1 \}$.
The DR-graph $S(\Theta)$ for $\Theta$ and $S$ is defined by the rules of Figure~\ref{figDRSwitchings}.
Then the following characterization holds.
\begin{theorem}[\cite{DR89}]
  \label{thm-DR-character}
  An MLL proof structure $\Theta$ is an MLL proof net
  if and only if
  for any DR-switching $S$ for $\Theta$, 
  the DR-graph $S(\Theta)$ is acyclic and connected. 
\end{theorem}
\begin{figure}[htbp]
\begin{center}
  \includegraphics[scale=0.55]{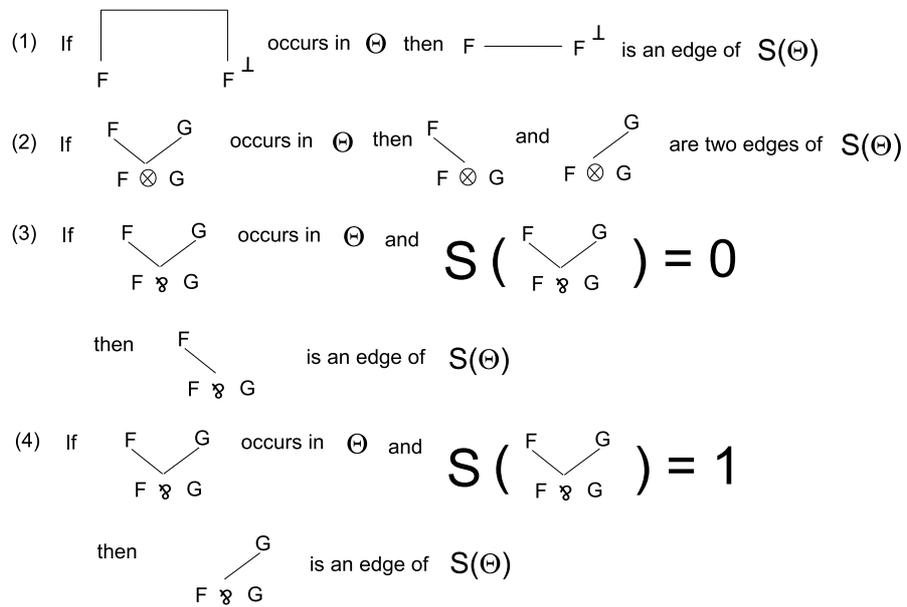}
\end{center}
 \caption{Definition of DR graphs}
 \label{figDRSwitchings}
\end{figure}
\subsection{de Naurois and Mogbil's correctness condition}
In this section we review de Naurois and Mogbil's correctness condition \cite{deNM11},
on which our linear time condition is based.
\begin{definition}
  A DR-switching $S$ for an MLL proof structure $\Theta$ is extreme left if
  for each $\PAR$-link $L$ in $\Theta$, $S$ always chooses the left premise in $L$.
  We denote the DR-switching by $S_{\forall \ell}$.
\end{definition}
In the following we only consider the extreme left switching.
We have no loss of generality under the assumption.
\begin{definition}
  \label{def-Consistency}
Let $\Theta$ be an MLL proof structure such that the DR-graph $S_{\forall \ell}(\Theta)$ is a tree. 
Let $L$ be a $\PAR$-link, and $n_{L}$, $n_{L}^{\ell}$ and $n_{L}^{r}$ be nodes in $S_{\forall \ell}(\Theta)$ induced by $L$, left, and right premises of $L$ respectively.
We say that $L$ is consistent in $S_{\forall \ell}(\Theta)$
if
the unique path $\theta$ from $n_{L}^{\ell}$ to $n_{L}^{r}$  in $S_{\forall \ell}(\Theta)$
does not contain $n_{L}$. 
\end{definition}
\begin{definition}
  \label{def-DirectedGraph}
  Let $\Theta$ be an MLL proof structure such that the DR-graph $S_{\forall \ell}(\Theta)$ is a tree
  and each $\PAR$-link in $\Theta$ is consistent in $S_{\forall \ell}(\Theta)$.
  Then we define a directed graph $G(S_{\forall \ell}(\Theta)) = (V, E)$ as follows:
  \begin{itemize}
  \item $V = \{ n_L \, | \, L \, \, \mbox{is a $\PAR$-link in} \, \, \Theta \}$
  \item 
    Let $L_1, L_2$ be different $\PAR$-links in $\Theta$.
    The directed edge $(n_{L_1}, n_{L_2})$ is in $E$ if
    the unique path from $n_{L_2}^{\ell}$ to $n_{L_2}^{r}$ in $S_{\forall \ell}(\Theta)$ contains
    the node $n_{L_1}$.
  \end{itemize}
\end{definition}

\begin{theorem}[\cite{deNM11}]
  \label{thm-deNM11}
  An MLL proof structure $\Theta$ is an MLL proof net iff
  \begin{enumerate}
  \item The DR-graph $S_{\forall \ell}(\Theta)$ is a tree.
  \item Each $\PAR$-link in $\Theta$ is consistent in $S_{\forall \ell}(\Theta)$. ($\PAR$-link consistency)
  \item The directed graph $G(S_{\forall \ell}(\Theta))$ is acyclic. (directed acyclicity)
  \end{enumerate}
\end{theorem}

\section{The Rewriting System over deNM-Trees}
In this section we introduce our rewriting system.
Then we give our correctness condition based on the system and
show that it is a characterization of MLL proof nets. 
\subsection{deNM-trees}
First we define deNM-trees.
In the following we fix an MLL proof structure $\Theta$ such that the DR-graph $S_{\forall \ell}(\Theta)$ is a tree.
\begin{definition}[deNM-trees]
  A deNM-tree is a finite tree consisting of labeled nodes and $\PAR$-nodes: 
  \begin{itemize}
  \item   A labeled node is labeled by a label set $S$
    whose each element is $l_L$ or $r_L$, where $L$ is a $\PAR$-link. 
    The degree $t$ of a labeled node is at most the number of nodes of the deNM-tree. 
    See Figure~\ref{figLabeledAndParNodes}.
  \item A $\PAR$-node is a labeled by a $\PAR$-link $L$.
    The degree of a $\PAR$-node is $1$ or $2$.
        See Figure~\ref{figLabeledAndParNodes}.
    As shown symbolically, we distinguish the port above of a $\PAR$-node from the port below. 
  \end{itemize}
  In a similar manner to Definition~\ref{def-Consistency}, we can define $\PAR$-consistency over deNM-trees.
  In addition, in a similar manner to Definition~\ref{def-DirectedGraph},
  a directed graph obtained from a deNM-tree and its acyclicity can be defined. 
\end{definition}
\begin{figure}[htbp]
\begin{center}
  \includegraphics[scale=0.6]{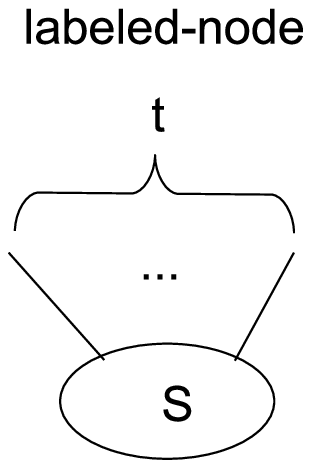}
  \includegraphics[scale=0.6]{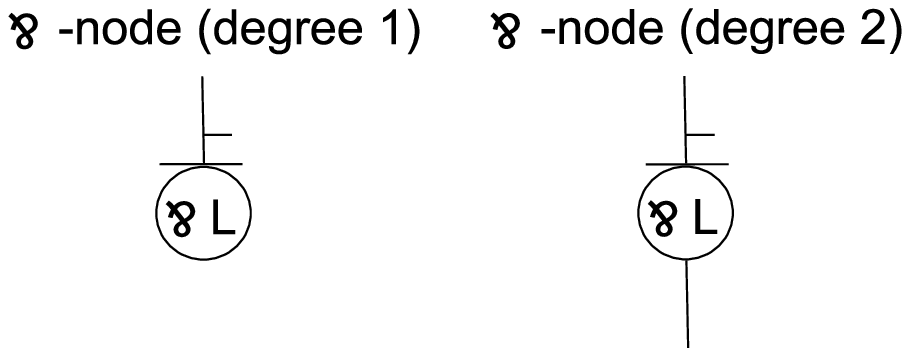}
\end{center}
 \caption{Labeled and $\PAR$-nodes}
 \label{figLabeledAndParNodes}
\end{figure}
Next we give a translation from $\Theta$ to a deNM-tree.
\begin{definition}
  We define a deNM-tree $T(\Theta)$ from $\Theta$ such that the DR-graph $S_{\forall \ell}(\Theta)$ is a tree as follows.
  If $\Theta$ consists of exactly one ID-link, then $T(\Theta)$ is a tree that consists of exactly one $0$-degree node labeled by $\emptyset$.
  Otherwise, for each link $L$ in $\Theta$ we specify a subtree $T_L$ in $T(\Theta)$ corresponding to $L$ as follows: 
\begin{itemize}
\item The case where $L$ is ID-link:
  \begin{enumerate}
    \item The case where one conclusion of $L$ is a right premise $F$ of a $\PAR$-link $L'$ or a conclusion $F$ of $\Theta$:
    Then $T_L$ consists of exactly one labeled node $n_L$ with degree $1$
    that is connected to the translation of the other conclusion $F^{\bot}$ of $L$
    (more precisely, $T_L$ is connected to the translation of the link whose left or right premise is $F^{\bot}$).
    Without loss of generality, we can assume that 
    $F^{\bot}$ is not a premise of $\PAR$-link 
    because otherwise, we can easily see that $\Theta$ is not an MLL proof net (in this case we define $T(\Theta)$ to be undefined).
    Then if the conclusion of $L$ is a right premise of $L'$, then the labeled set of $n_L$ is $\{ r_{L'} \}$.
    Otherwise, that of $n_L$ is empty.
    See Figure~\ref{figTranslatonID-linkCase1}.
    \begin{figure}[htbp]
      \begin{center}
        \includegraphics[scale=0.55]{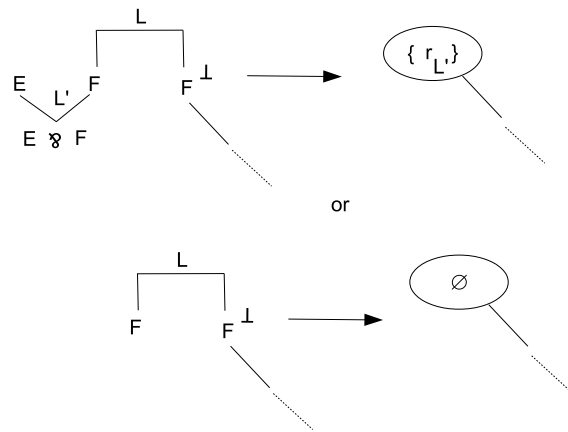}
      \end{center}
      \caption{ID-link (Case 1)}
      \label{figTranslatonID-linkCase1}
    \end{figure}
  \item Otherwise:
    In this case
    without loss of generality we can assume that one of the conclusions of $L$ is not a premise of a $\PAR$-link
    because when both conclusions of $L$ are a premise of a $\PAR$-link, we can easily see that $\Theta$ is not an MLL proof net (in this case we define $T(\Theta)$ to be undefined).
    Then $T_L$ consists of exactly one labeled node $n_L$ with degree $2$.
    If one of the conclusions of $L$ is a left premise of a $\PAR$-link $L'$, then
    the labeled set for $n_L$ is $\{ \ell_{L'} \}$.
    Otherwise the labeled set for $n_L$ is $\emptyset$.
    See Figure~\ref{figTranslatonID-linkCase2}.
    \begin{figure}[htbp]
      \begin{center}
        \includegraphics[scale=0.55]{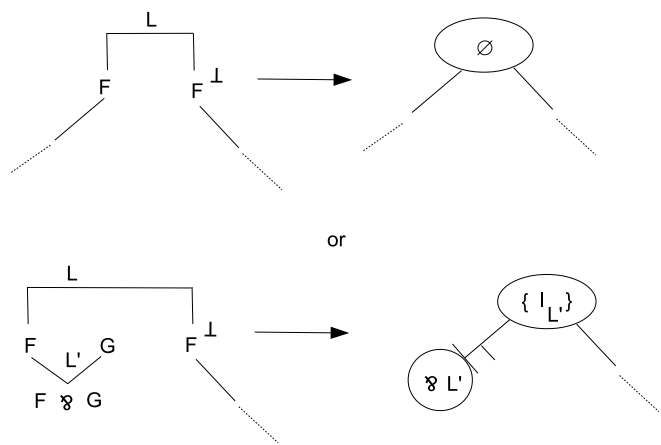}
      \end{center}
      \caption{ID-link (Case 2)}
      \label{figTranslatonID-linkCase2}
    \end{figure}
  \end{enumerate}
\item The case where $L$ is $\TENS$-link:
  \begin{enumerate}
  \item The case where the conclusion of $L$ is a conclusion of $\Theta$ or a right premise of a $\PAR$-link $L'$:
    In this case $T_L$ consists of exactly one labeled node $n_L$ with degree $2$
    that is connected to trees translated from both premises of $L$.
    If the conclusion of $L$ is a right premise of $L'$, then the labeled set for $n_L$ is $\{ r_{L'} \}$.
    Otherwise, the labeled set for $n_L$ is $\emptyset$.
    See Figure~\ref{figTranslatonTENS-linkCase1}.
    \begin{figure}[htbp]
      \begin{center}
        \includegraphics[scale=0.55]{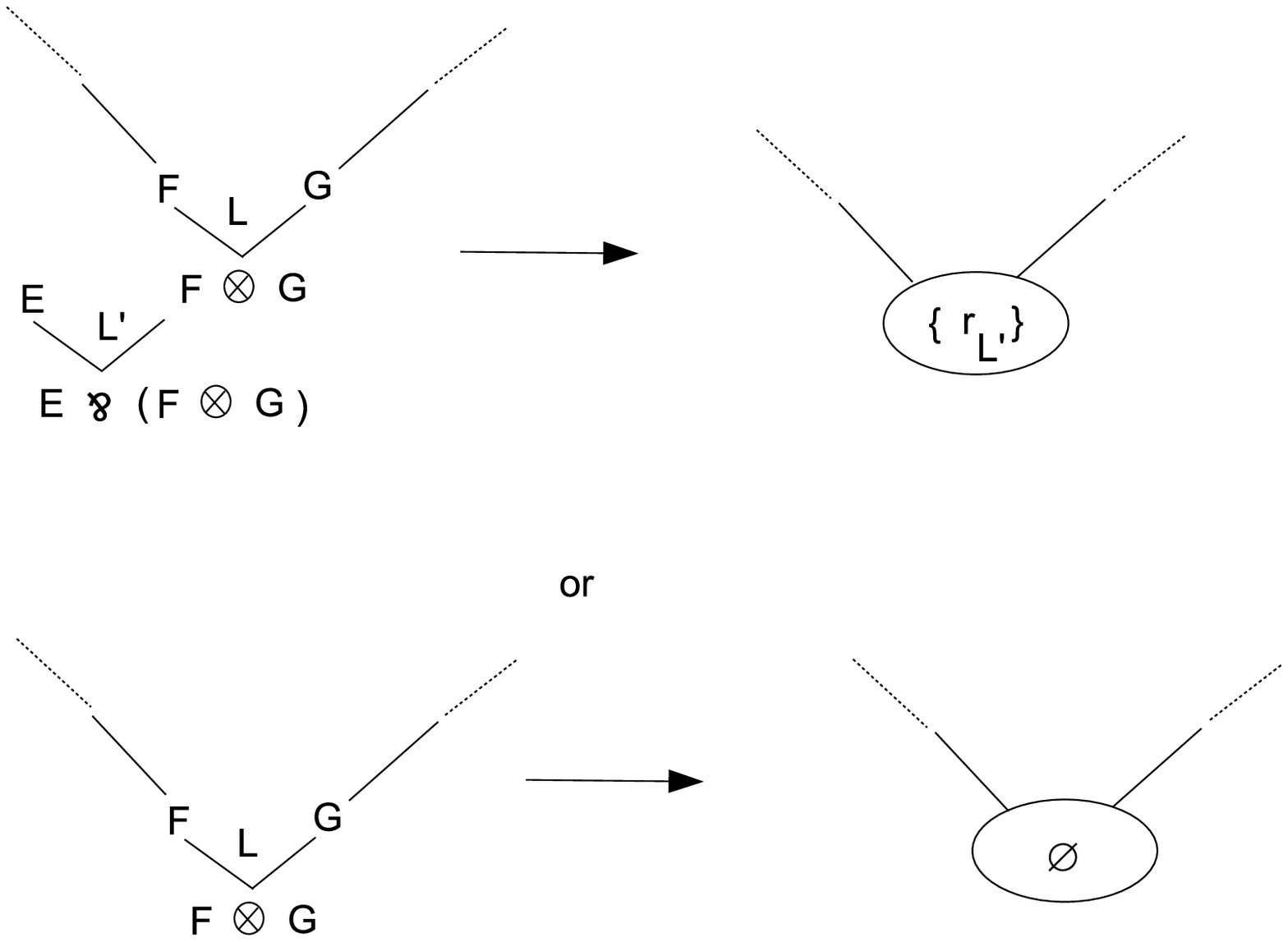}
      \end{center}
      \caption{$\TENS$-link (Case 1)}
      \label{figTranslatonTENS-linkCase1}
    \end{figure}
  \item Otherwise:
    In this case $T_L$ consists of exactly one labeled node $n_L$ with degree $3$.
    If the conclusion of $L$ is a left premise of a $\PAR$-link $L'$, then the labeled set for $n_L$ is $\{ \ell_{L'} \}$.
    Otherwise, the labeled set for $n_L$ is $\emptyset$.
    See Figure~\ref{figTranslatonTENS-linkCase2}.
    \begin{figure}[htbp]
      \begin{center}
        \includegraphics[scale=0.55]{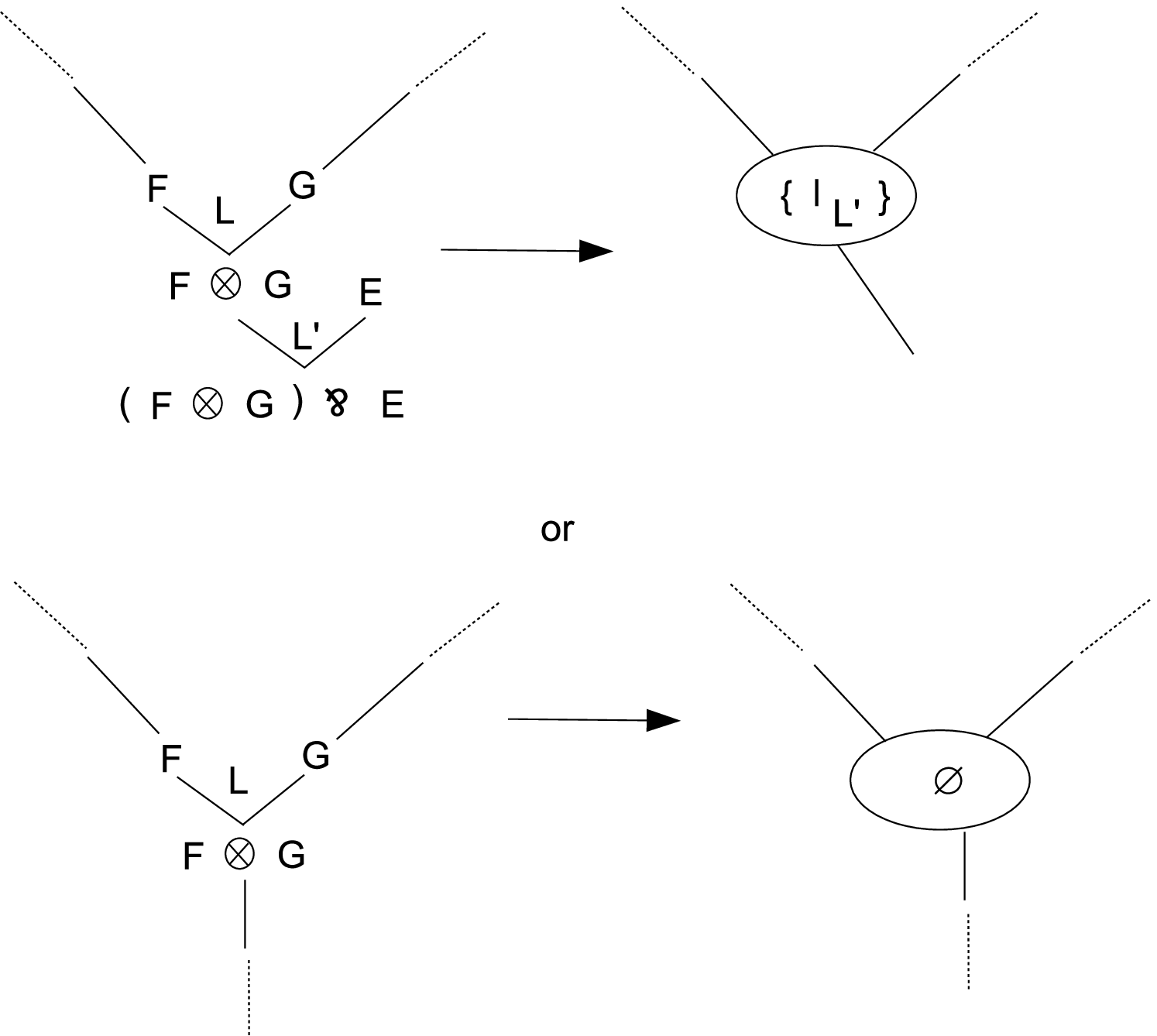}
      \end{center}
      \caption{$\TENS$-link (Case 2)}
      \label{figTranslatonTENS-linkCase2}
    \end{figure}
  \end{enumerate}
\item The case where $L$ is $\PAR$-link:
  \begin{enumerate}
  \item The case where the conclusion of $L$ is a right premise of a $\PAR$-link $L'$:
    In this case $T_L$ consists of one labeled node $n_1$ with degree $1$ labeled by $\{ r_{L'} \}$ and one $\PAR$-node $n_L$ labeled by $L$ with degree $2$
    such that $n_1$ and $n_L$ is connected.
    The node $n_L$ is connected to the tree translated from the left premise of $L$. 
    See Figure~\ref{figTranslatonPAR-linkCase1}.
      \begin{figure}[htbp]
        \begin{center}
          \includegraphics[scale=0.55]{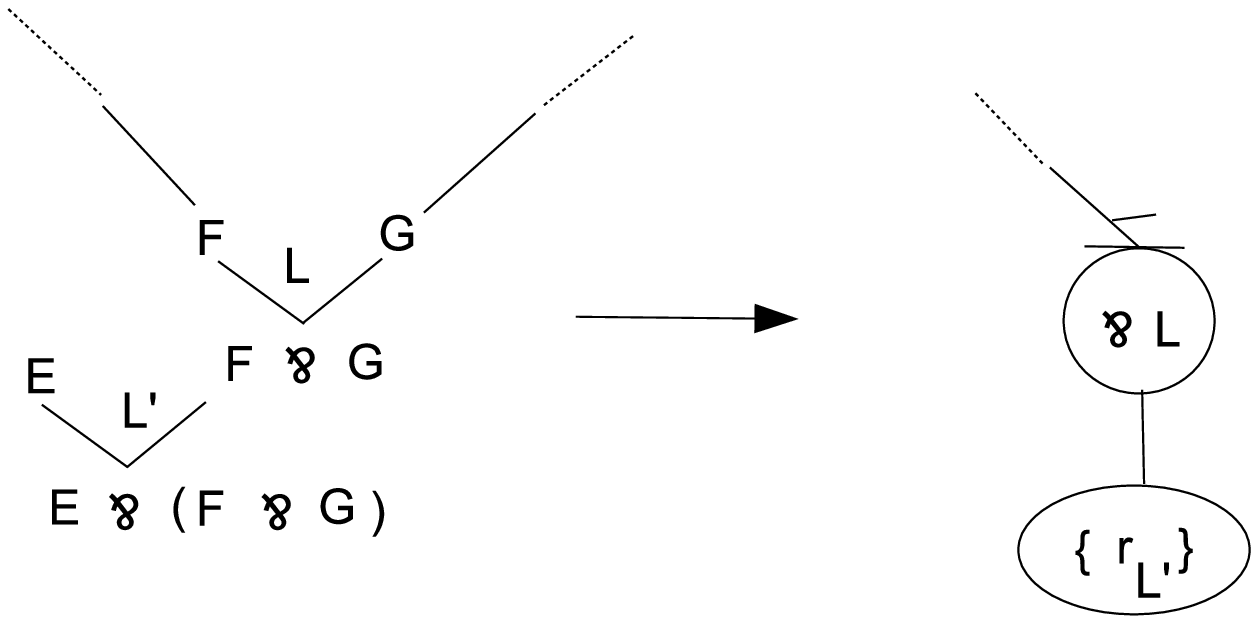}
        \end{center}
        \caption{$\PAR$-link (Case 1)}
        \label{figTranslatonPAR-linkCase1}
      \end{figure}
  \item 
    The case where the conclusion of $L$ is a left premise of a $\PAR$-link $L'$:
    In this case $T_L$ consists of one labeled node $n_2$ with degree $2$ labeled by $\{ \ell_{L'} \}$ and one $\PAR$-node $n_L$ labeled by $L$ with degree $2$
    such that $n_2$ and $n_L$ is connected.
    While $n_L$ is connected to the tree translated from the left premise of $L$,
    $n_2$ is connected to the tree translated from the conclusion of $L$.
    See Figure~\ref{figTranslatonPAR-linkCase2}.
    \begin{figure}[htbp]
      \begin{center}
        \includegraphics[scale=0.55]{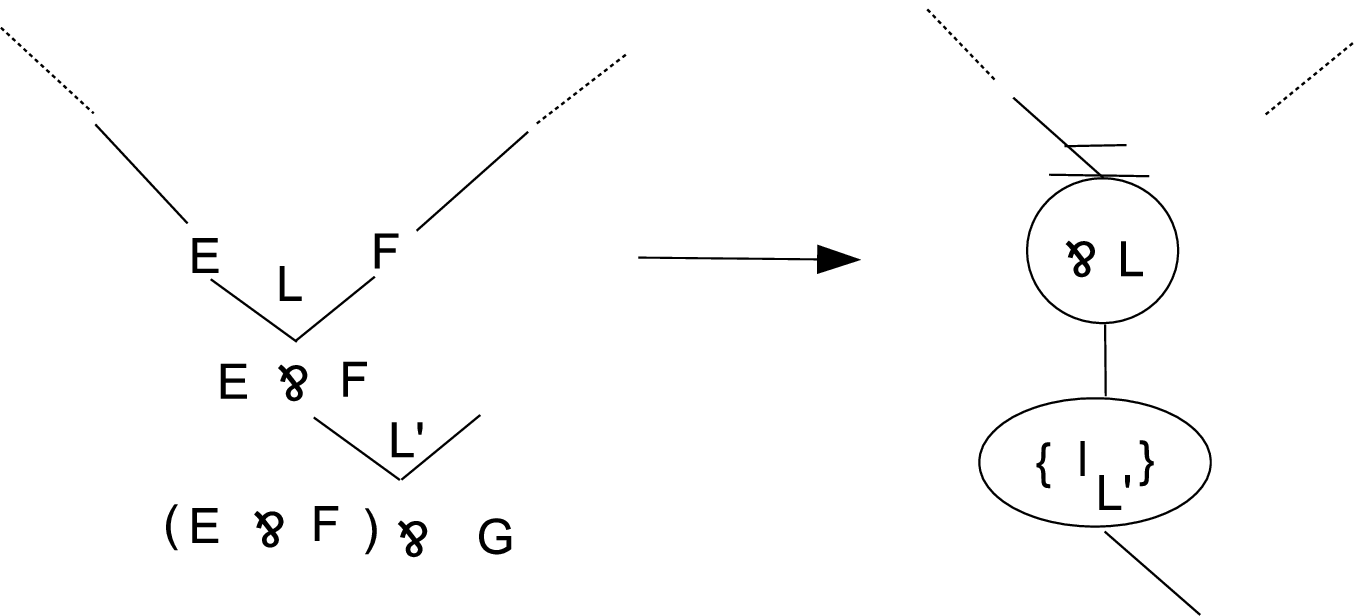}
      \end{center}
      \caption{$\PAR$-link (Case 2)}
      \label{figTranslatonPAR-linkCase2}
    \end{figure}
  \item Otherwise:
    In this case $T_L$ consists of exactly one $\PAR$ node $n_L$ labeled by $L$.
    If $L$ is a conclusion of $\Theta$, then the degree of $n_L$ is $1$.
    Otherwise, the degree of $n_L$ is $2$.
    See Figure~\ref{figTranslatonPAR-linkCase3}.
    \begin{figure}[htbp]
      \begin{center}
        \includegraphics[scale=0.55]{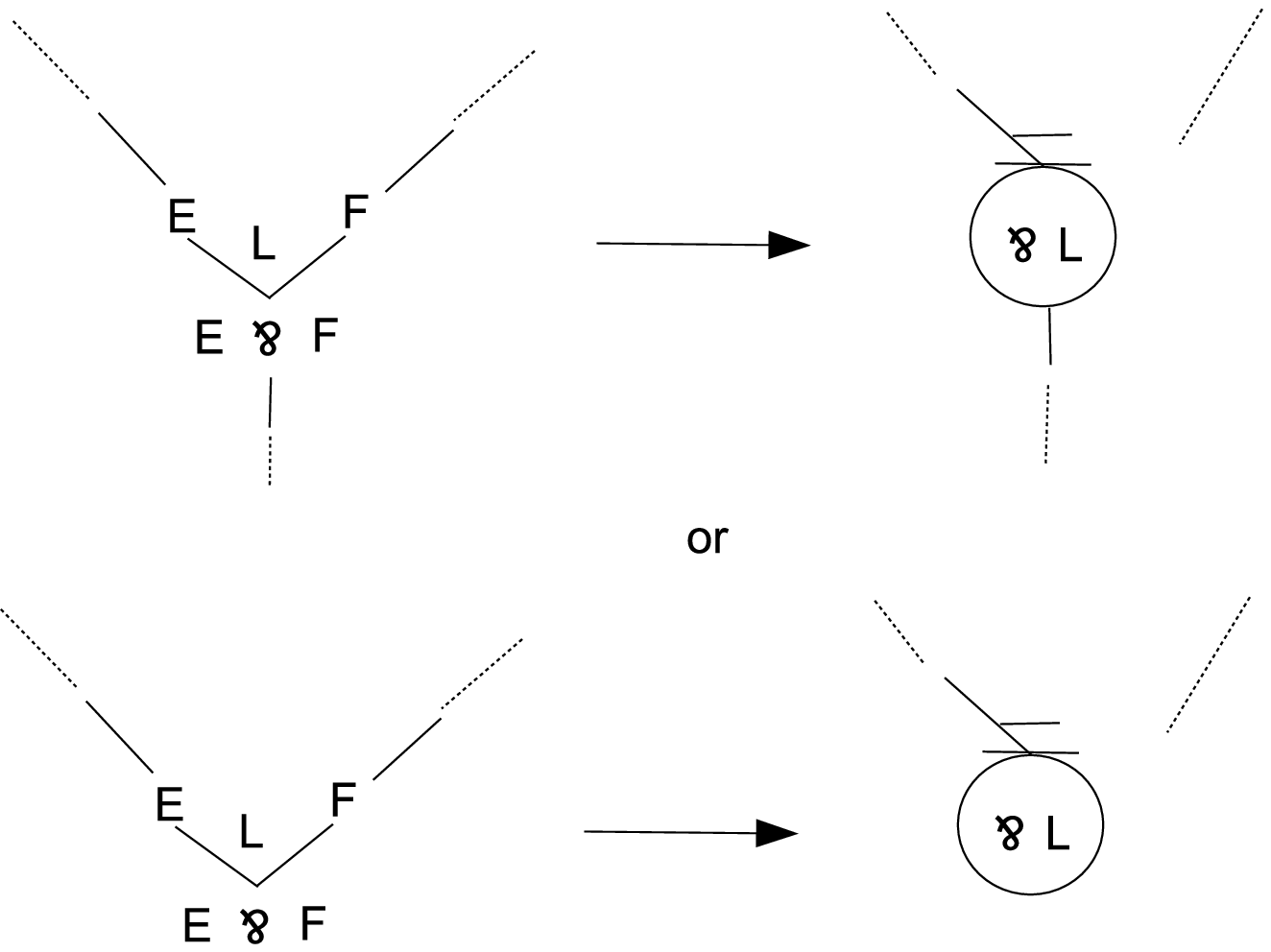}
      \end{center}
      \caption{$\PAR$-link (Case 3)}
 \label{figTranslatonPAR-linkCase3}
\end{figure}
  \end{enumerate}
\end{itemize}
Then $T(\Theta$) is the tree obtained by connecting these subtrees $T_L$.
\end{definition}
If $T(\Theta)$ is defined, then we can easily see that $T(\Theta)$ is a deNM-tree
because we assume that $S_{\forall \ell}(\Theta)$ is a tree.
\subsection{The Rewriting System over deNM-Trees}
Next we introduce our rewriting system over deNM-trees.
In the rewriting system we must specify exactly one node in a deNM-tree that is about to be rewritten, which we call the {\it active} node in the deNM-tree.
The active node must be a labeled node. 
Our rewriting system has only three rewrite rules. 
\begin{itemize}
\item The rewrite rule of Figure~\ref{figTranslatonRewriteRuleParElim} is called $\PAR$-{\it elimination}:
  If the active node $n$ is connected to a $\PAR$-node $n_L$ labeled by $L$ through the port above 
  and 
  the label set $S$ of $n$ contains labels $\ell_L$ and $r_L$, then $n_L$ is eliminated. 
  \begin{figure}[htbp]
    \begin{center}
      \includegraphics[scale=0.55]{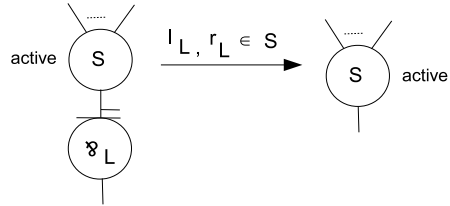}
    \end{center}
    \caption{$\PAR$-elimination rule}
    \label{figTranslatonRewriteRuleParElim}
  \end{figure}
  \item The rewrite rule of Figure~\ref{figTranslatonRewriteRuleUnion} is called {\it union}:
    If the active node is connected to a labeled node, then these two nodes are merged.
    The label set of the resulting node is the union of them of merged two nodes. 
    \begin{figure}[htbp]
      \begin{center}
        \includegraphics[scale=0.55]{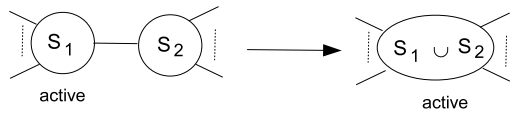}
      \end{center}
      \caption{Union rule}
      \label{figTranslatonRewriteRuleUnion}
    \end{figure}
  \item
    The rewrite rule called {\it local jump} of Figure~\ref{figTranslatonRewriteRuleLocalJump} 
    does not change any nodes:
    It just changes the current active node.
    Note that in this rewrite rule, the active node before the rewrite is connected to a $\PAR$-node $L$ through the port below
    and the active node after the rewrite is the labeled node whose label set contains $r_{L}$. 
    \begin{figure}[htbp]
      \begin{center}
        \includegraphics[scale=0.55]{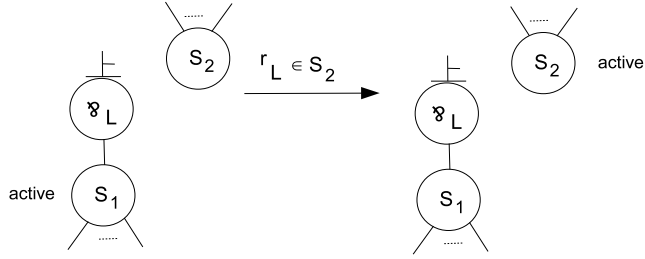}
      \end{center}
      \caption{Local jump rule}
      \label{figTranslatonRewriteRuleLocalJump}
    \end{figure}
\end{itemize}
We denote the rewriting system consisting of the three rewrite rules above by $\MCAL{R}$. 
\begin{proposition}
  \label{prop-deNM-tree-basic-1}
  Let $\Theta$ be an MLL proof structure such that $S_{\forall \ell}(\Theta)$ is a tree. 
  Then $\Theta$ is an MLL proof net iff $T(\Theta)$ satisfies $\PAR$-consistency and directed acyclicity for deNM-trees.
\end{proposition}
\begin{proof}
    It is obvious from Theorem~\ref{thm-deNM11}.
    $\Box$
\end{proof}
\begin{proposition}
  \label{prop-deNM-tree-basic-2}
  Let $T$ be an deNM-tree.
  Then let $T'$ be an deNM-tree obtained from $T$ by choosing one active node $n$ and applying one of three rewrite rules to $n$.
  \begin{enumerate}
  \item [(a)] If $T$ satisfies $\PAR$-consistency and directed acyclicity for deNM-trees, then $T'$ also satisfies them.
  \item [(b)] If $T$ does not satisfy $\PAR$-consistency, then $T'$ does not satisfy the property.
  \item [(c)] If $T$ does not satisfy directed acyclicity, then $T'$ does not satisfy the property.
  \end{enumerate}
\end{proposition}
\begin{proof}
  \begin{enumerate}
  \item [(a)] Each rewrite rule preserves $\PAR$-consistency and directed acyclicity.
  \item [(b)] Each rewrite rule preserves $\PAR$-inconsistency.
    An inconsistent $\PAR$-link can not be removed by the $\PAR$-elimination rule. 
  \item [(c)] Each rewrite rule cannot cancel directed cyclicity.
  \end{enumerate}
\end{proof}
For example, let $\Theta_0$ be the MLL proof structure shown in Figure~\ref{fignonPN-0}, where 
the symbol $\circledcirc$ means a $\PAR$-link occurrence.
Then $\Theta_0$ is not an MLL proof net
because $T(\Theta_0)$ shown in Figure~\ref{figdeNMTree-0} satisfies $\PAR$-consistency, but not directed acyclicity.
Moreover whatever we choose any labeled node as the starting active node,
we cannot cancel directed cyclicity by applying one of three rewrite rules. 
\begin{figure}[htbp]
\begin{center}
  \includegraphics[scale=0.5]{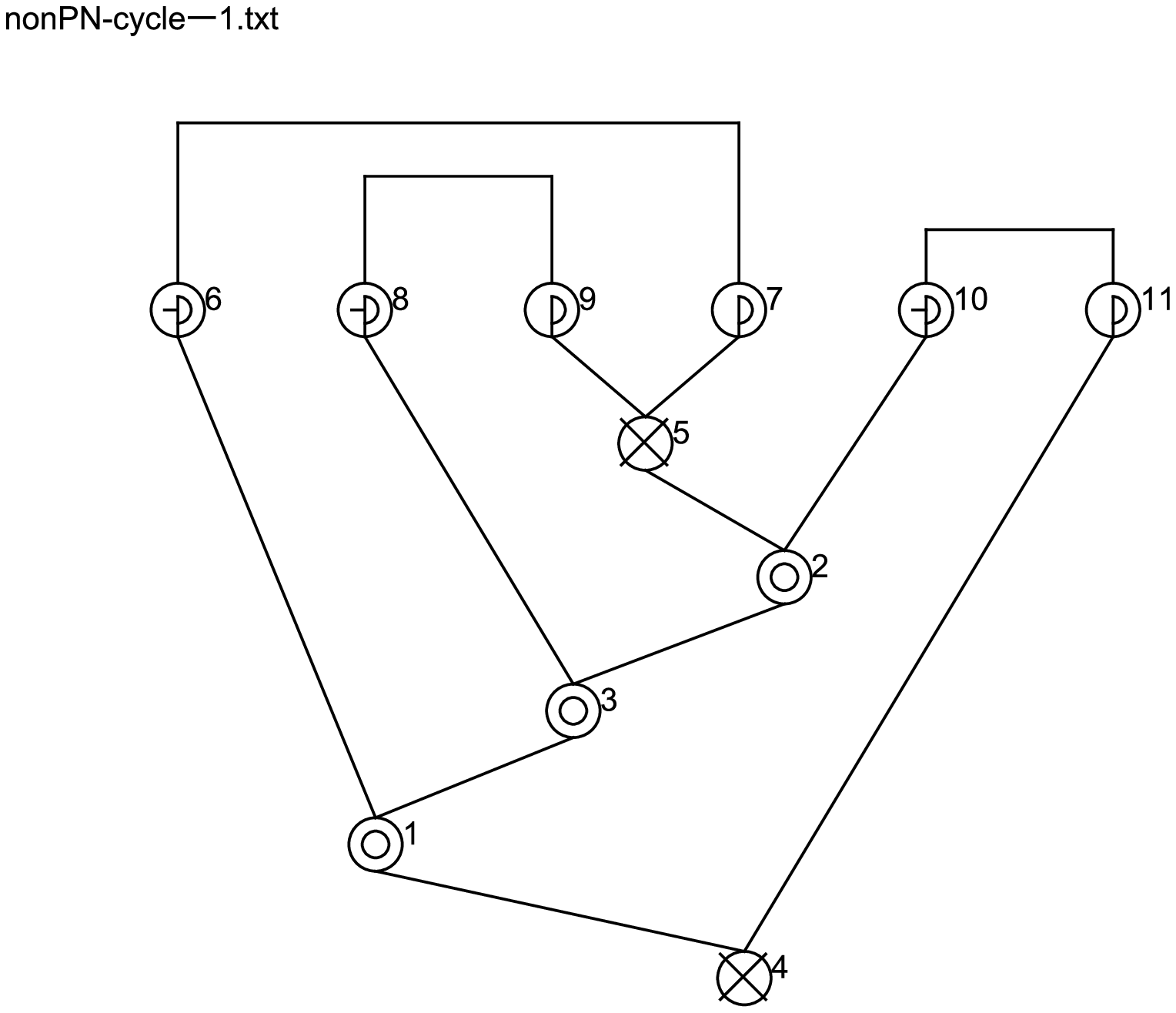}
\end{center}
 \caption{MLL Proof Structure $\Theta_0$, but not MLL Proof Net}
 \label{fignonPN-0}
\end{figure}
\begin{figure}[htbp]
\begin{center}
  \includegraphics[scale=0.5]{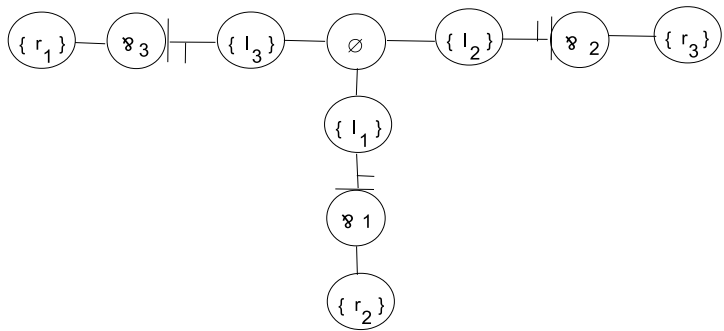}
\end{center}
 \caption{deNM-tree $T(\Theta_0)$, which does not satisfy directed acyclicity}
 \label{figdeNMTree-0}
\end{figure}

Let $\Theta$ be an MLL proof structure 
and $\MBB{L}_{\Theta} = \{ L_1, \ldots, L_m \}$ be the set of all $\PAR$-links in $\Theta$.
Then we define the full label set $S_{\rm full}$ to be
\[
S_{\rm full} = \{ \ell_L \, | \, L \in \MBB{L}_{\Theta} \} \cup \{  r_L \, | \, L \in \MBB{L}_{\Theta} \}
\]
\begin{definition}
  Algorithm $A$ is defined as follows:

    \begin{tabular}{ll}
    & {\bf Input}: an MLL proof structure $\Theta$\\
    & {\bf Output}: {\bf yes} or {\bf no}. \\
    1.& If the deNM-tree $T(\Theta)$ is not defined, then the output is {\bf no}. \\
      & Otherwise go to 2. \\
    2.& A labeled node $n$ in $T(\Theta)$ is selected arbitrarily. \\
    3.& Rewriting is started with $T(\Theta)$ and the active node $n$ using three \\
      & rewrite rules above. \\
    4.& If the local jump rule is applied to a $\PAR$-link to which the local \\
      & jump rule has been applied already, then the output is {\bf no}.\\
    5.& When any of three rewrite rules cannot be applied to the current \\
      & deNM-tree $T'$, if $T'$ consists of exactly one node labeled by \\
      & $S_{\rm full}$ with degree $0$, then the output is {\bf yes}. \\
      & Otherwise, the output is {\bf no}.
  \end{tabular}
\end{definition}
\begin{proposition}
  \label{prop-termination}
  Algorithm $A$ always terminates. 
\end{proposition}
\begin{proof}
  Algorithm $A$ cannot be applied the local jump rule to a $\PAR$-link more than one time.
  Both of the other two rules reduce the number of nodes in a deNM-tree. 
  $\Box$
\end{proof}
\begin{lemma}
  \label{lemma-Main}
  If Algorithm $A$ terminates in Step 5, then $\Theta$ is not an MLL proof net. 
\end{lemma}
\begin{proof}
  In this case, $T(\Theta)$ must reduce to a deNM-tree with configuration shown in Figure~\ref{figMainLemma-1}.
  Then if $\Theta$ does not violate the second condition of Theorem~\ref{thm-deNM11}, i.e., the $\PAR$-link consistency, then 
  the configuration must extend to the configuration shown in Figure~\ref{figMainLemma-2}.
  But it is not a tree anymore since it has a cycle.
  Therefore $\Theta$ is not an MLL proof net. 
  $\Box$
\end{proof}
    \begin{figure}[htbp]
      \begin{center}
        \includegraphics[scale=0.55]{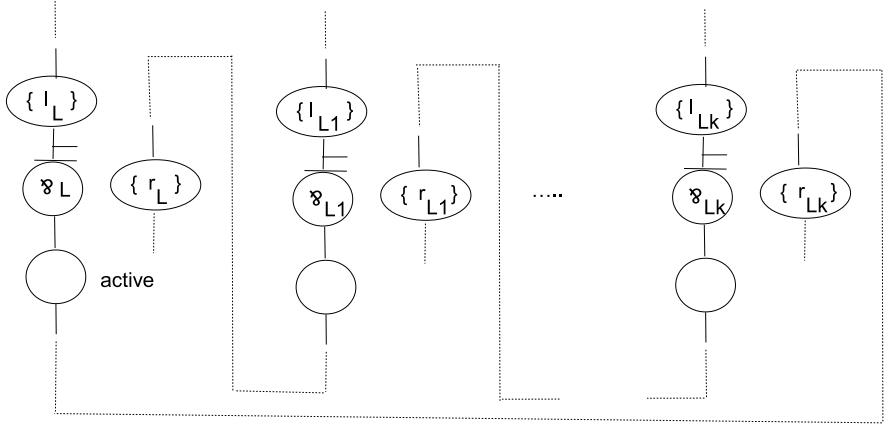}
      \end{center}
      \caption{Configuration for Termination at step 5}
      \label{figMainLemma-1}
    \end{figure}
    \begin{figure}[htbp]
      \begin{center}
        \includegraphics[scale=0.55]{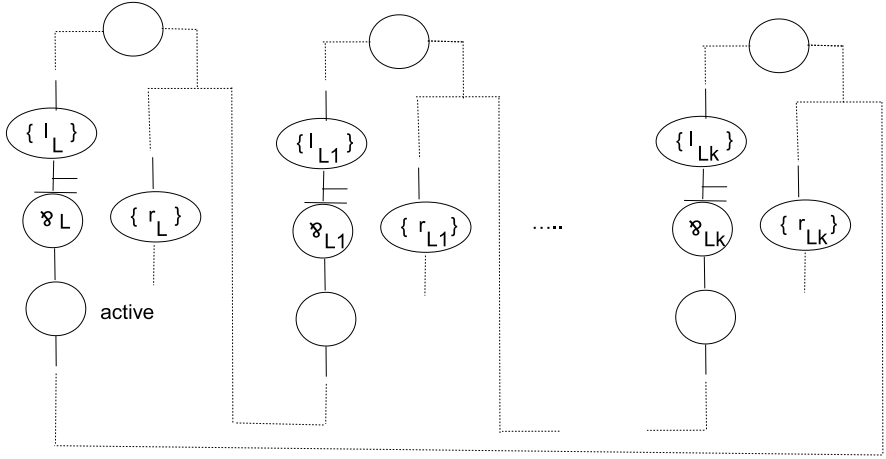}
      \end{center}
      \caption{But not a deNM-tree anymore}
      \label{figMainLemma-2}
    \end{figure}

\begin{theorem}
  \label{mainTheorem}
  Let $\Theta$ be an MLL proof structure.
  Then
  $\Theta$ is an MLL proof net if and only if Algorithm $A$ with input $\Theta$ outputs {\bf yes}.
\end{theorem}
\begin{proof}
  \begin{itemize}
  \item Only-if-part:
    By Proposition~\ref{prop-termination}, 
    Algorithm $A$ terminates. 
    Hence we can suppose that $\Theta$ is an MLL proof net and Algorithm $A$ with input $\Theta$ outputs {\bf no}.
    If the deNM-tree $T(\Theta)$ is not well-defined in Step 1, then it means that $S_{\forall \ell}(\Theta)$ is not a tree
    and contradicts Theorem~\ref{thm-DR-character}.
    Moreover application of the local jump rule to a $\PAR$-link twice in Step 4 means that
    we can find a DR-switching $S$ such that $S(\Theta)$ is not a tree.
    It also contradicts Theorem~\ref{thm-DR-character}.
    So Algorithm $A$ reaches Step 5.
    But it contradicts that Lemma~\ref{lemma-Main}. 
    So Algorithm $A$ must terminates with exactly one node tree with degree $0$.
    Moreover, the node must be labeled by $S_{\rm full}$.
  \item If-part:
    We suppose that Algorithm $A$ with input $\Theta$ outputs {\bf yes}.
    Then $\Theta$ automatically satisfies the first condition of Theorem~\ref{thm-deNM11}.
    That is, the deNM-tree $T(\Theta)$ must be well-defined. 
    We suppose that there is an inconsistent $\PAR$-link $L$ in $S_{\forall \ell}(\Theta)$.
    Then our rewriting system cannot be reduced $T(\Theta)$ to one node tree,
    because we cannot apply the $\PAR$-elimination rule to the $\PAR$-node $n_L$. 
    So $\Theta$ satisfies the second condition of Theorem~\ref{thm-deNM11}, i.e., the $\PAR$-link consistency. 
    We suppose that the directed graph $G(S_{\forall \ell}(\Theta))$ has a cycle.
    Then our rewriting system cannot be reduced $T(\Theta)$ to one node tree,
    because we cannot apply the $\PAR$-elimination rule to the $\PAR$-links which are contained in the cycle. 
    Hence $\Theta$ satisfies the third condition of Theorem~\ref{thm-deNM11}.
    Therefore $\Theta$ must be an MLL proof net. 
  $\Box$
  \end{itemize}
\end{proof}
\subsection{Examples}
We show three examples in this section.
Figure~\ref{figPN-1} shows an MLL proof net $\Theta_1$, where
the symbol $\circledcirc$ means a $\PAR$-link occurrence.
This figure has been generated using the Proof Net Calculator \cite{Mat19a}.
It is translated to the deNM-tree $T(\Theta_1)$ shown in Figure~\ref{figPN-1-Tree}.
When you choose any labeled node as the starting active node,
you must finally reach to one labeled node with degree $0$ labeled by the full label set
\[
S_{\rm full} = \{ \ell_1, r_1, \ell_2, r_2, \ell_3, r_3, \ell_4, r_4 \} 
\]
using our three rewrite rules.

\begin{figure}[htbp]
\begin{center}
  \includegraphics[scale=0.38]{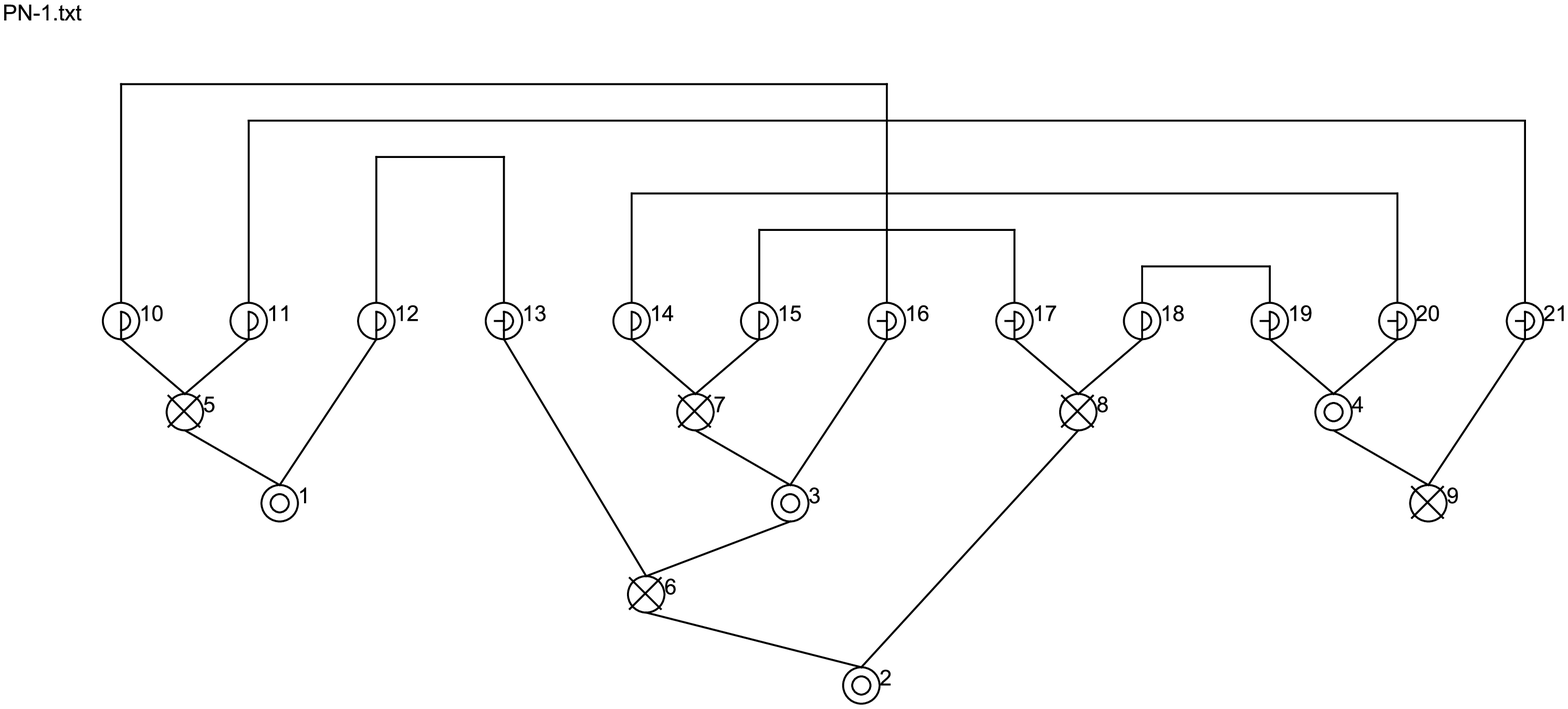}
\end{center}
 \caption{MLL Proof Net $\Theta_1$}
 \label{figPN-1}
\end{figure}
\begin{figure}[htbp]
\begin{center}
  \includegraphics[scale=0.45]{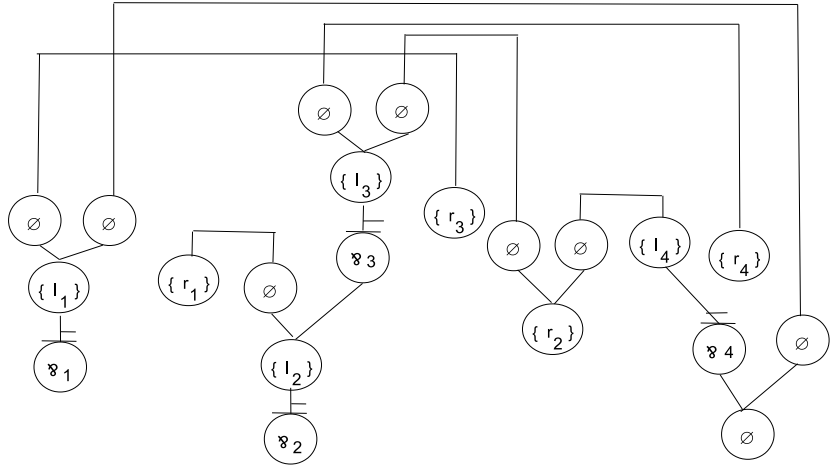}
\end{center}
 \caption{$T(\Theta_1)$}
 \label{figPN-1-Tree}
\end{figure}
Figure~\ref{fignonPN-1} shows an MLL proof structure $\Theta_2$ that is not an MLL proof net.
It is translated to the deNM-tree $T(\Theta_2)$ shown in Figure~\ref{fignonPN-1-Tree}.
When you choose any labeled node as the starting active node,
you cannot reach to one labeled node with degree $0$ labeled by the full label set
\[
S_{\rm full} = \{ \ell_1, r_1, \ell_2, r_2, \ell_3, r_3, \ell_4, r_4 \} 
\]
using our three rewrite rules.
\begin{figure}[htbp]
\begin{center}
  \includegraphics[scale=0.34]{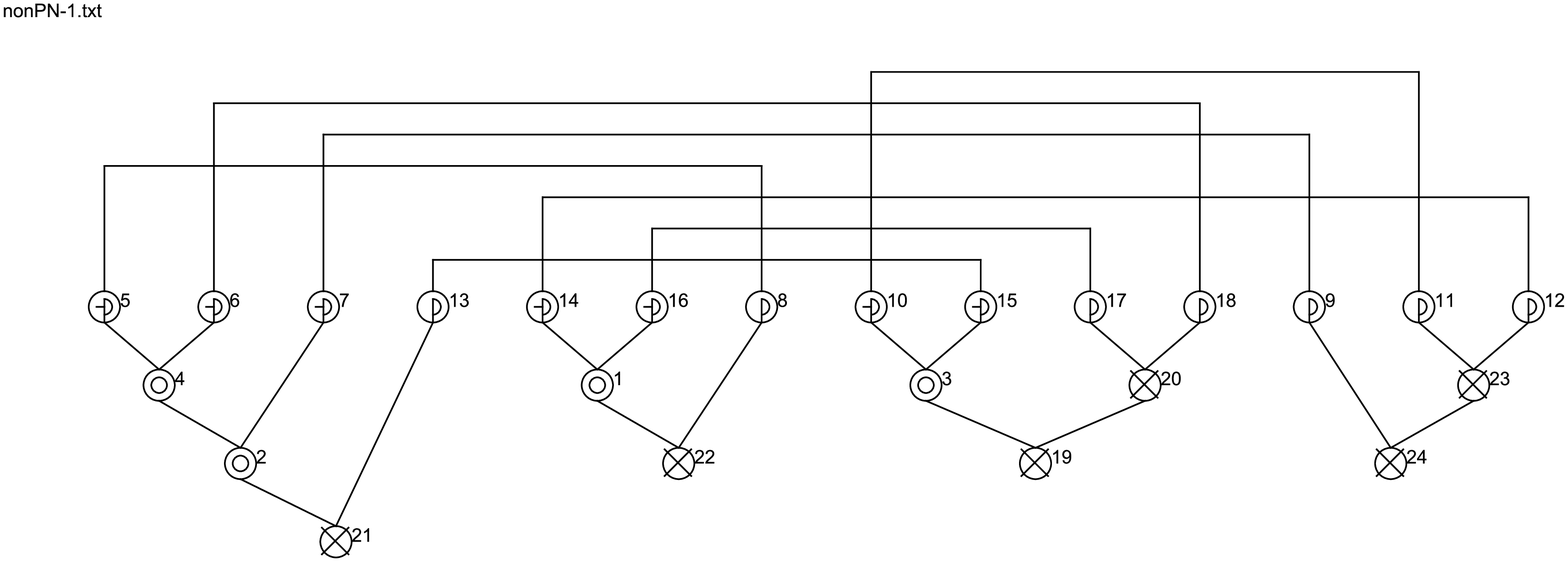}
\end{center}
 \caption{MLL Proof Structure $\Theta_2$, but not MLL Proof Net}
 \label{fignonPN-1}
\end{figure}
\begin{figure}[htbp]
\begin{center}
  \includegraphics[scale=0.42]{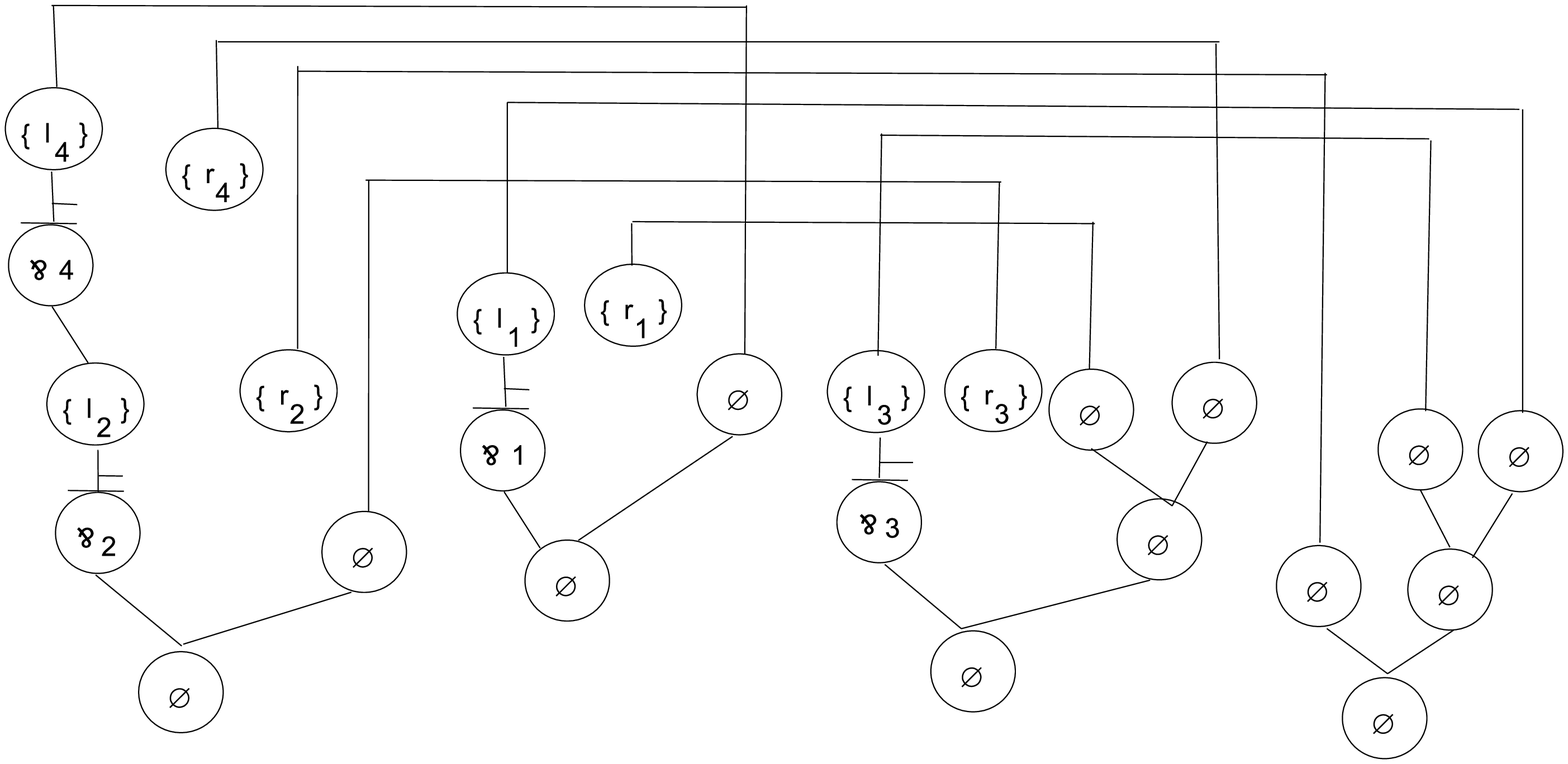}
\end{center}
 \caption{$T(\Theta_2)$}
 \label{fignonPN-1-Tree}
\end{figure}

Figure~\ref{fignonPN-loop-1} shows an MLL proof structure $\Theta_3$ that is not an MLL proof net.
It is translated to the deNM-tree $T(\Theta_3)$ shown in Figure~\ref{fignonPN-loop-1-Tree}.
When we choose the node labeled by $\{ r_2 \}$ as the starting rule,
the first rewrite rule may be the local jump rule for $\PAR$-link $1$.
Then the node labeled by $\{ r_2 \}$ becomes active.
After two applications of the union rule, the local jump rule for $\PAR$-link $1$ must be tried to be applied again.
Then step 5 in Algorithm $A$ outputs {\bf no}.
\begin{figure}[htbp]
\begin{center}
  \includegraphics[scale=0.5]{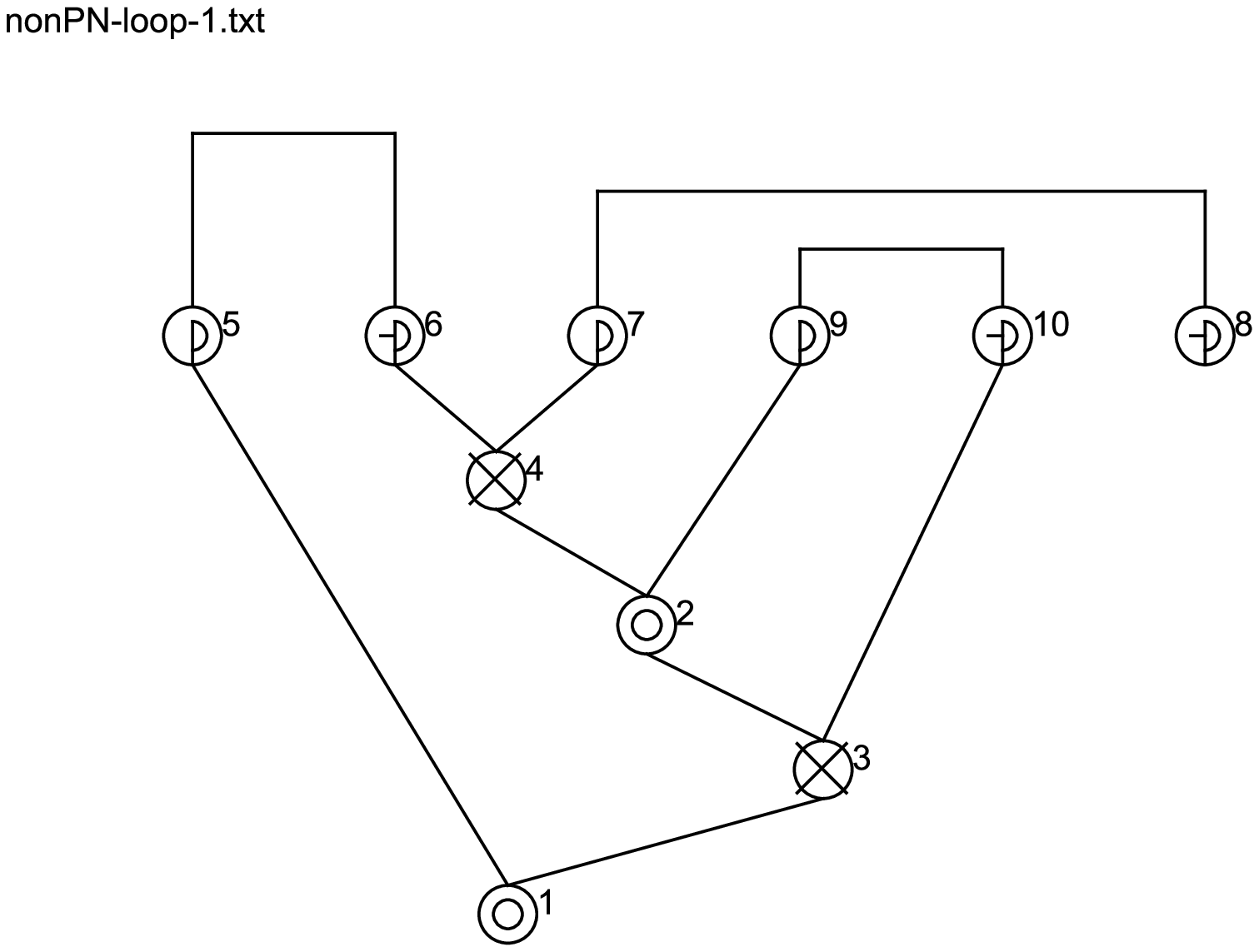}
\end{center}
 \caption{MLL Proof Structure $\Theta_3$, but not MLL Proof Net}
 \label{fignonPN-loop-1}
\end{figure}
\begin{figure}[htbp]
\begin{center}
  \includegraphics[scale=0.42]{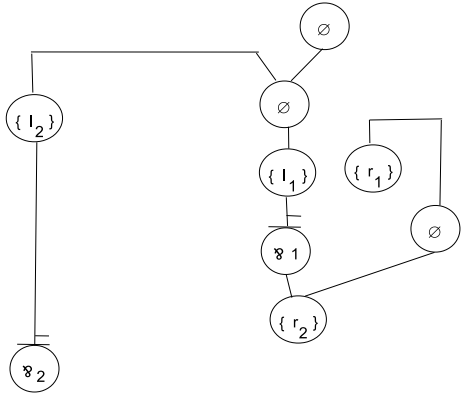}
\end{center}
 \caption{$T(\Theta_3)$}
 \label{fignonPN-loop-1-Tree}
\end{figure}

\section{Linear Time Correctness Condition}
Although our rewriting system $\MCAL{R}$ is surprisingly simple,
it cannot establish linear time termination,
because a node in a deNM-tree $T$ may have a degree depending on the number of nodes of $T$
and therefore take quadratic time in the worst case. 
For example, reduction of the deNM-tree shown in Figure~\ref{figQuadraticEx1} to one node tree 
may take quadratic time $O(k^2)$ in $\MCAL{R}$
because before each application of the union rule it may try to apply the $\PAR$-elimination rule
to the active node and $\PAR$-node $i \, (1 \le i \le k)$.
\begin{figure}[htbp]
\begin{center}
  \includegraphics[scale=0.5]{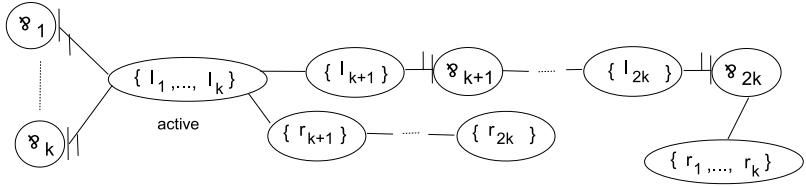}
\end{center}
 \caption{A deNM-Tree that may needs quadratic computation at the worst case scenario}
 \label{figQuadraticEx1}
\end{figure}
In order to establish linear time termination based on our rewriting system,
we must restrict a way of application of rewrite rules using more sophisticated data structures.

Let $\Theta$ be an MLL proof structure and
$T$ be a deNM-tree occurring during reduction, which starts from $T(\Theta)$. 
We assume that each labeled node $n$ in $T$ has the following data structures:
\begin{itemize}
\item The queue $\MCAL{Q}_{\rm down}$ of $\PAR$-nodes connected to $n$ from the port below.
\item The queue $\MCAL{Q}_{\rm labeled}$ of labeled-nodes connected to $n$.
\item The queue $\MCAL{Q}_{\rm right}$ of right premise labels $r_L$ included in the label set on $n$
  that have not been tried for $\PAR$-elimination yet
  or, have been put into $\MCAL{S}_{\rm right}$ once but have been put again by ``revival'' mechanism.
  Initially, if the label node is labeled by $r_L$, then $\MCAL{Q}_{\rm right} = r_L$.
  Otherwise, $\MCAL{Q}_{\rm right}$ is empty. 
\item The set $\MCAL{S}_{\rm right}$ of right premise labels $r_L$ included in the label set on $n$
  that have already been tried for $\PAR$-elimination, where
  the set is a partition in a disjoint set-union data structure \cite{CLRS09}.
  Initially, $\MCAL{S}_{\rm right}$ is always empty.
  We call $\MCAL{S}_{\rm right}$ the {\it right premise label set} for $n$.
\item The queue $\MCAL{Q}_{\rm up}$ of $\PAR$-nodes connected to $n$ from the port above that have not been tried for $\PAR$-elimination yet
  or, have been put into $\MCAL{S}_{\rm up}$ once but have been put again by ``revival'' mechanism..
\item The set $\MCAL{S}_{\rm up}$ of $\PAR$-nodes connected to $n$ from the port above
  and have already been tried for $\PAR$-elimination, where
  the set is a partition in a disjoint set-union data structure \cite{CLRS09}.
  We call $\MCAL{S}_{\rm up}$ the {\it up port set} for $n$.
  Initially, $\MCAL{S}_{\rm up}$  is always empty.
\item The set $\MCAL{S}_{\cup {\rm right}}$, whose purpose is to avoid deadlock of $\PAR$-elimination application.
  Initially $\MCAL{S}_{\cup {\rm right}}$ is empty.
  The set $\MCAL{S}_{\cup {\rm right}}$ behaves like $\MCAL{S}_{{\rm right}}$  initially,
  but in the case of application of the local jump rule,
  $\MCAL{S}_{\cup {\rm right}}$ in the previous active node is merged into that of the new active node. 
  We call $\MCAL{S}_{\cup {\rm right}}$ the {\it merged right premise label set} for $n$. 
\item The set $\MCAL{S}_{\cup {\rm up}}$, whose purpose is to avoid deadlock of $\PAR$-elimination application.
  Initially $\MCAL{S}_{\cup {\rm up}}$ is empty.
  The set $\MCAL{S}_{\cup {\rm up}}$ behaves like $\MCAL{S}_{{\rm up}}$  initially,
  but in the case of application of the local jump rule,
  $\MCAL{S}_{\cup {\rm up}}$ in the previous active node is merged into that of the new active node. 
  We call $\MCAL{S}_{\cup {\rm up}}$ the {\it merged up port set} for $n$. 
\end{itemize}
In the initial stage, we can associate these data structures to each labeled node in $T(\Theta)$ in linear time. 

Let the current active node in $T$ be $n_{\rm act}$.
Then we define our reduction strategy as follows:
\begin{enumerate}
\item First if  $\MCAL{Q}_{\rm down}$ for $n_{\rm act}$ is not empty, then
  the local jump rule is applied to $n_{\rm act}$ and the first element of $\MCAL{Q}_{\rm down}$.
  Before the application, the first element is deleted from $\MCAL{Q}_{\rm down}$.
  This deletion can be done in constant time.
  Moreover $\MCAL{S}_{\cup {\rm up}}$ and $\MCAL{S}_{\cup {\rm right}}$ for the previous active node
  is merged into $\MCAL{S}_{\cup {\rm up}}$ and $\MCAL{S}_{\cup {\rm right}}$ for the new active node respectively. 
\item Second in the case where  $\MCAL{Q}_{\rm down}$ for $n_{\rm act}$ is empty and $\MCAL{Q}_{\rm labeled}$ for $n_{\rm active}$ is not empty, 
  if the first element $n'$ of $\MCAL{Q}_{\rm labeled}$ does not denote itself, i.e, $n_{\rm act}$, then 
  the union rule is applied to $n_{\rm act}$ and $n'$.
  After the application, the data structures for two nodes $n_{\rm act}$ and $n'$ are merged in such a way that 
  $n'$ is deleted from $\MCAL{Q}_{\rm labeled}$.
  These merges can be done in constant time. 
  Otherwise, i.e., if the first element of $\MCAL{Q}_{\rm labeled}$ denote itself, then
  the element is deleted from $\MCAL{Q}_{\rm labeled}$ and return to the beginning of this step.
\item Third in the case where both $\MCAL{Q}_{\rm down}$ and $\MCAL{Q}_{\rm labeled}$ for $n_{\rm act}$ are empty,
  if $n_{\rm act}$ is one node tree, then the output is {\bf yes}.
  Otherwise, one of the following cases is tried to be applied in order: 
  \begin{itemize}
  \item The case where both $\MCAL{Q}_{\rm right}$ and $\MCAL{Q}_{\rm up}$ are empty: Then the output is {\bf no}.
  \item The case where $\MCAL{Q}_{\rm right}$ is not empty:
    Then  
    let the first element of $\MCAL{Q}_{\rm right}$ be $r_L$.
    If $\MCAL{S}_{\rm up}$ includes $\PAR$-link $L$,
    then the $\PAR$-elimination rule is applied to $n_{\rm act}$ and $L$.
    Then $r_L$ is deleted from $\MCAL{Q}_{\rm right}$ and
    when $\PAR$-link $L$ has the labeled node $n'$ connected to the port below,
    $n'$ is appended to $\MCAL{Q}_{\rm labeled}$ for $n_{\rm act}$. 
    These operations can be done in constant time.
    Otherwise, i.e., if $\MCAL{S}_{\rm up}$ does not includes $\PAR$-link $L$, then
    there are the following two cases:
    \begin{itemize}
    \item The case where $L$ is included in $\MCAL{S}_{\cup {\rm up}}$:
      In this case, there must be a labeled node $n'$ in which
      $L$ is put into $\MCAL{S}_{\cup {\rm up}}$.
      Then $L$ is put into the queue $\MCAL{Q}_{\rm up}$ for the labeled node which integrates $n'$ in the current deNM-tree.
      These operations can be done in constant time.
    \item Otherwise:
      $r_L$ is deleted from $\MCAL{Q}_{\rm right}$ and put in $\MCAL{S}_{\rm right}$ and $\MCAL{S}_{\cup {\rm right}}$.
      These operations can be done in constant time.
    \end{itemize}
  \item The case where $\MCAL{Q}_{\rm up}$ is not empty:
    Then 
    let the first element of $\MCAL{Q}_{\rm up}$ be $\PAR$-link $L$.
    If $\MCAL{S}_{\rm right}$ includes $r_L$, then
    then we apply the $\PAR$-elimination rule to $n_{\rm act}$ and $L$.
    Then $L$ is deleted from new $\MCAL{L}_{\rm up}$ and 
    when $\PAR$-link $L$ has the labeled node $n'$ connected to the port below,
    $n'$ is appended to $\MCAL{Q}_{\rm labeled}$ for $n_{\rm act}$. 
    These operations can be done in constant time.
    Otherwise, i.e., if $\MCAL{S}_{\rm right}$ does not includes $r_L$, then
    there are the following two cases:
    \begin{itemize}
    \item The case where $r_L$ is included in $\MCAL{S}_{\cup {\rm right}}$:
      In this case, there must be a labeled node $n'$ in which
      $r_L$ is put into $\MCAL{S}_{\cup {\rm right}}$.
      Then $r_L$ is put into the queue $\MCAL{Q}_{\rm right}$ for the labeled node which integrates $n'$ in the current deNM-tree.
      These operations can be done in constant time.
    \item Otherwise:
      $L$ is deleted from $\MCAL{Q}_{\rm up}$ and put in $\MCAL{S}_{\rm up}$ and $\MCAL{S}_{\cup {\rm up}}$.
      These operations can be done in constant time.
    \end{itemize}
  \end{itemize}
\end{enumerate}
\begin{definition}
We call the modified Algorithm $A$ with the data structures and the strategy above Algorithm $B$. 
\end{definition}
\begin{remark}
  \begin{itemize}
  \item When the local jump rule is applied, the first element of $\MCAL{Q}_{\rm down}$ is deleted.
    Therefore unlike the rewriting system $\MCAL{R}$,
    we do not need the second application check of the local jump rule to the same  $\PAR$-link anymore.
    But in order to detect a cycle in a DR-graph as soon as possible, this check may be included. 
  \item In order to establish linear time termination, we can not maintain the set of $\PAR$-nodes connected to the active node from the port above
    as a sole {\it queue data structure}.
    For example, if the active node in deNM-tree shown in Figure~\ref{figQuadraticEx2-1} maintain the information as
    the queue $1, 2, \ldots, k$, then
    we must scan the queue and delete one element at each $\PAR$-elimination, so that
    the reduction to one node tree takes quadratic time at the worst case. 
  \item In order to establish linear time termination, it is essential to adopt a {\it disjoint set-union data structure} \cite{CLRS09}.
    By using not only the disjoint set-union data structure $\MCAL{S}_{\rm up}$ but also a queue data structure $\MCAL{Q}_{\rm up}$,
    the amortized cost becomes linear. 
  \item By the similar reason, we use a disjoint set-union data structure $\MCAL{S}_{\rm right}$ and a queue data structure $\MCAL{Q}_{\rm right}$
    in order to maintain the set of right premise labels on the active node.
  \item We does neither delete the eliminated $\PAR$-link from $\MCAL{S}_{\rm up}$ nor the right premise label $r_L$ from $\MCAL{S}_{\rm right}$,
    because the cost may be linear,
    implying quadratic time termination at the worst case and the deletion is not necessary.
  \item We need to merge $\MCAL{S}_{\cup {\rm right}}$ and $\MCAL{S}_{\cup {\rm up}}$
    from the previous active node to that of the new active node in the local jump rule.
    Moreover we need to ``revival'' of $\PAR$-link labels and right premise labels from $\MCAL{S}_{\rm up}$ and $\MCAL{S}_{\rm right}$
    to $\MCAL{Q}_{\rm up}$ and $\MCAL{Q}_{\rm right}$ respectively. 
    For example, let us consider
    the deNM-tree $T(\Theta)$ shown in Figure~\ref{fig-deNMTree-RSVD-1} from obtained from the MLL proof net $\Theta$ shown in Figure~\ref{figRSVD-1}.
    We suppose that we don't have such cares.
    Then starting with $T(\Theta)$, after applications of $\PAR$-elimination for $\PAR_2$ and $\PAR_3$,
    $r_1$ and $\PAR_1$ would be in $\MCAL{S}_{\rm right}$ and $\MCAL{S}_{\rm up}$ for the active node respectively.
    Moreover $\MCAL{Q}_{\rm right}$ and $\MCAL{Q}_{\rm up}$ would be empty.
    This means that we are in deadlock, so we can not eliminate $\PAR_1$.
    That's why we need the ``revival'' mechanism. 
  \end{itemize}
\begin{figure}[htbp]
\begin{center}
  \includegraphics[scale=0.5]{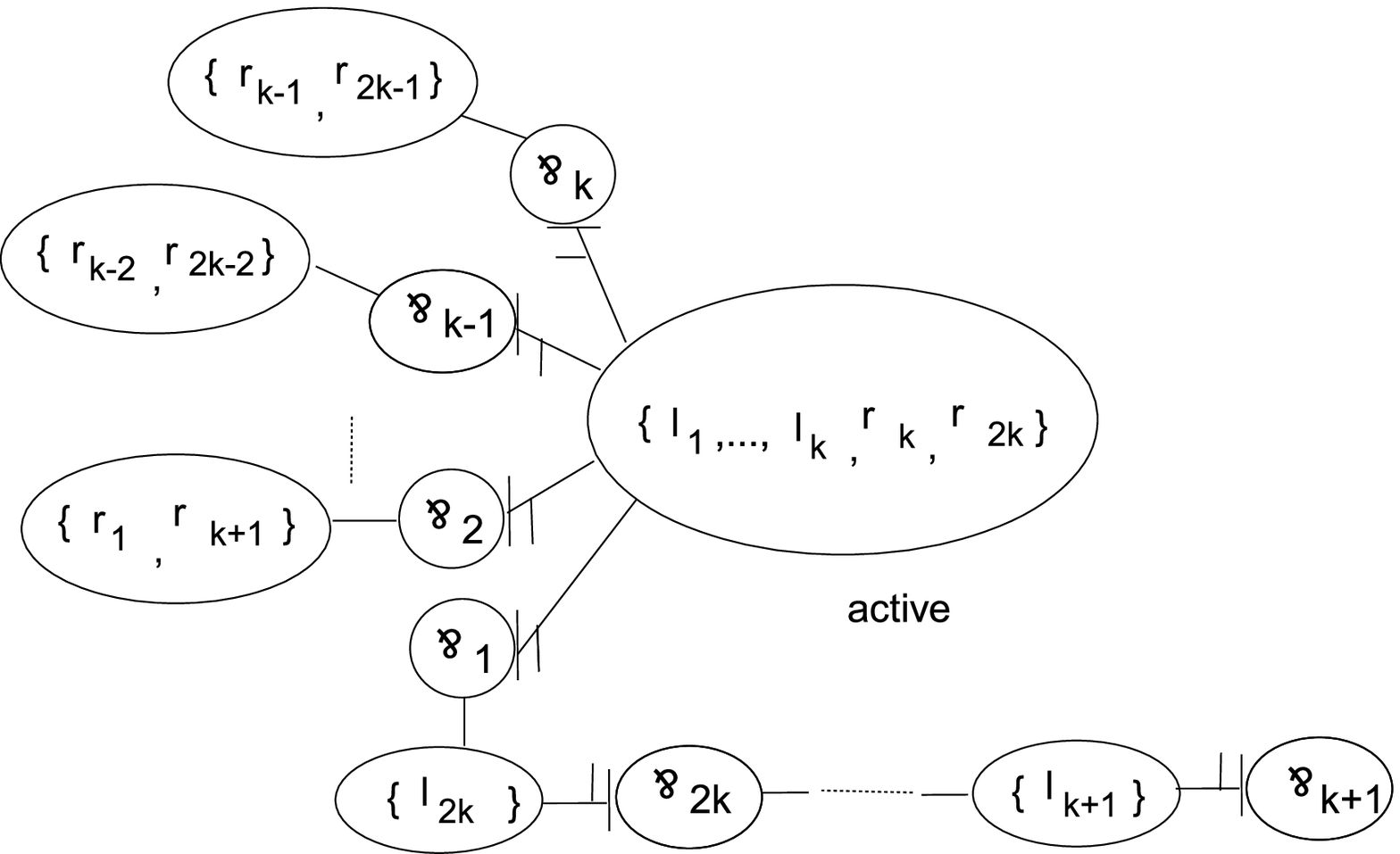}
\end{center}
 \caption{A deNM-Tree that needs quadratic computation at the worst case scenario}
 \label{figQuadraticEx2-1}
\end{figure}
\begin{figure}[htbp]
\begin{center}
  \includegraphics[scale=0.5]{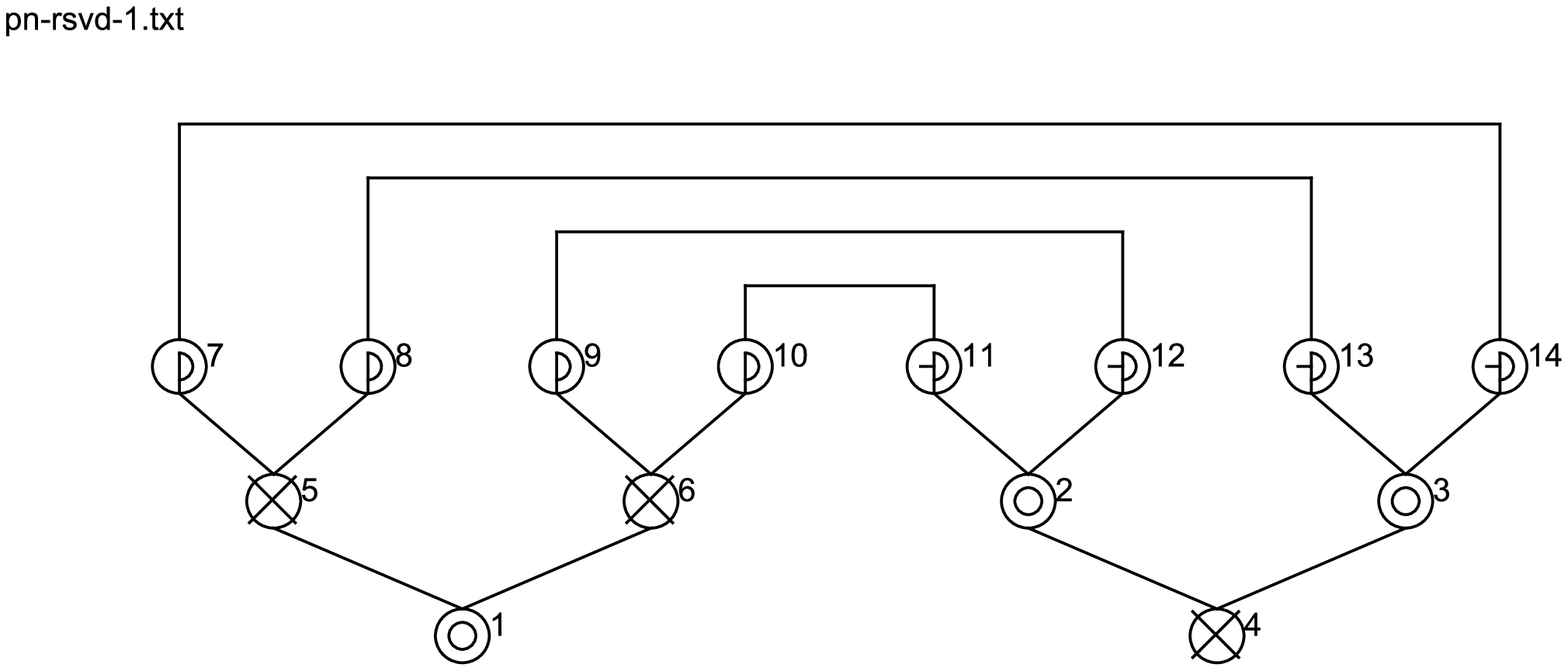}
\end{center}
 \caption{An MLL proof net}
 \label{figRSVD-1}
\end{figure}
\begin{figure}[htbp]
\begin{center}
  \includegraphics[scale=0.5]{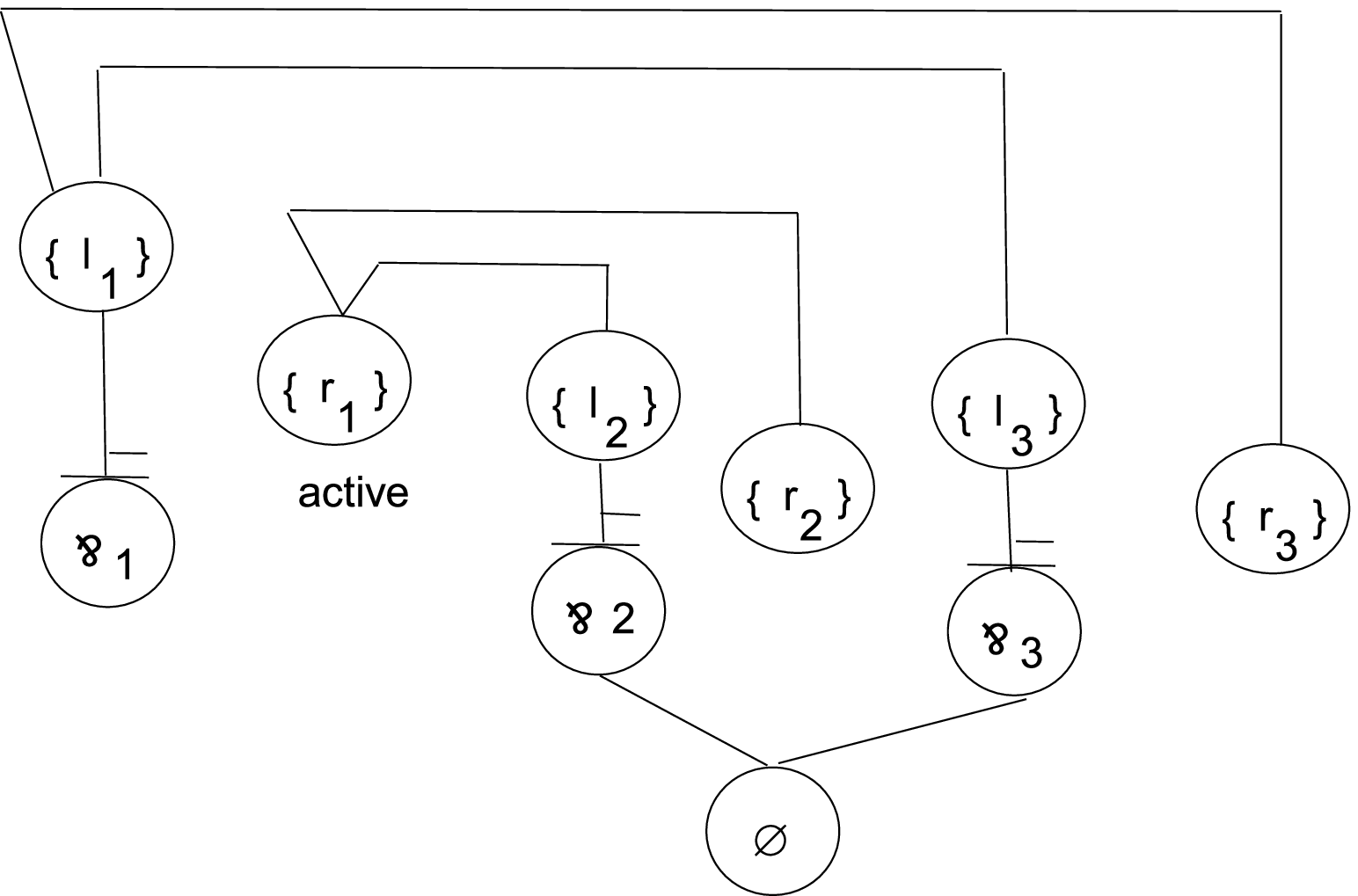}
\end{center}
 \caption{A deNM-Tree that needs $\MCAL{S}_{\cup {\rm right}}$ and $\MCAL{S}_{\cup {\rm up}}$ for linear time termination}
 \label{fig-deNMTree-RSVD-1}
\end{figure}
\end{remark}
\begin{proposition}
Algorithm $B$ terminates.
\end{proposition}
\begin{proof}
  We use measure $(m_1 + m_2 + m_3 + m_4 , n_1 + n_2)$ for the proof, where
  \begin{enumerate}
  \item $m_1$ is the number of labeled nodes and $\PAR$-nodes in the deNM-tree.
  \item $m_2$ is the number of $\PAR$-nodes that have not been visited yet by the local jump rule.
  \item $m_3$ is the number of $\PAR$-nodes that is included in $\MCAL{Q}_{\rm up}$ for a labeled node and have been included from the starting point.
  \item $m_4$ is the number of right premise labels that is included in $\MCAL{Q}_{\rm right}$ for a labeled node and have been included from the starting point.
  \item $n_1$ is the number of $\PAR$-nodes that have not been included in $\MCAL{Q}_{\rm up}$ for a labeled node initially,
    but have been put in it by ``revival mechanism''.
  \item $n_2$ is the number of right premise labels that have not been included in $\MCAL{Q}_{\rm right}$ for a labeled node initially
    but have been put in it by ``revival mechanism''.
  \end{enumerate}
  Then it is easily see that
  the measure $(m, n)$ strictly decreases in lexicographic order for each step.
  $\Box$
\end{proof}
\begin{proposition}
  \begin{itemize}
  \item If $\Theta$ is an MLL proof net, then Algorithm $B$ with the input $\Theta$ answers {\bf yes}.
  \item If $\Theta$ is not an MLL proof net, then Algorithm $B$ with the input $\Theta$ answers {\bf no}.
  \end{itemize}
\end{proposition}
\begin{proof}
  Each step in Algorithm $B$ rewrite a deNM-tree using one of three rules or let the deNM-tree remain.
  So by Theorem~\ref{mainTheorem} the latter statement holds.
  We observe that for each step in Algorithm $B$ 
  the following property holds:
  \begin{quote}
    For each $\PAR$-node $L$ in the deNM-tree $T$,
    $L$ is included in $\MCAL{Q}_{\rm up}$
    or
    $r_L$ is included in $\MCAL{Q}_{\rm right}$ for a labeled node in $T$.
  \end{quote}
  Then if $\Theta$ is an MLL proof net, then
  Algorithm $B$ with input $\Theta$ answers {\bf yes}
  since
  the above property guarantees that each $\PAR$-node $L$ in $T(\Theta)$ is eliminated. 
  $\Box$
\end{proof}
\begin{theorem}
  Let $\Theta$ be an MLL proof structure. 
  There is a random access machine that simulate Algorithm $B$ with input $\Theta$ in $O(n)$ time, where
  $n$ is the number of the links in $\Theta$.
\end{theorem}
\begin{proof}
  The tree check for $S_{\forall \ell}(\Theta)$ can be computed in $O(n)$
  by using breadth-first or depth-first search.
  If $S_{\forall \ell}(\Theta)$ is a tree, then
  $T(\Theta)$ can also be obtained and an arbitrary labeled node in $T(\Theta)$ be found in $O(n)$.
  In the tree rewriting process, each node can be visited at most twice.
  The total number of applications of three rewrite rules is $O(n)$. 
  The only nontrivial part is management of right premise label sets $\MCAL{S}_{\rm right}$ (and $\MCAL{S}_{\cup {\rm right}}$)
  and that of up port sets $\MCAL{S}_{\rm up}$ (and $\MCAL{S}_{\cup {\rm up}}$). 
  The union operation in the union rule and the query (find) operation in the $\PAR$-elimination and local jump rules 
  are a typical instance of disjoint-set union-find operations \cite{CLRS09}.
  We note that we only use a fixed number of these operations in each rewrite step. 
  There is not only a query operation between $\MCAL{S}_{\rm up}$ (and $\MCAL{S}_{\cup {\rm up}}$) and
  the first element in $\MCAL{Q}_{\rm right}$ 
  but also that between $\MCAL{S}_{\rm right}$ (and $\MCAL{S}_{\cup {\rm right}}$) and the first element in $\MCAL{Q}_{\rm up}$.
  In general case, the amortized cost of these operations is superlinear.
  However, if the underlying structure is a tree known in advance and the union operation is only performed between one node and its parent,
  then the amortized cost is linear in the total number of both operations in the random access machine model \cite{GT85}. 
  Luckily this applies to our case.
  Therefore our claim holds. 
  $\Box$
\end{proof}

\begin{example}
  We consider the MLL proof net $\Theta_4$ shown in Figure~\ref{figPN2}.
  \begin{figure}[htbp]
\begin{center}
  \includegraphics[scale=0.5]{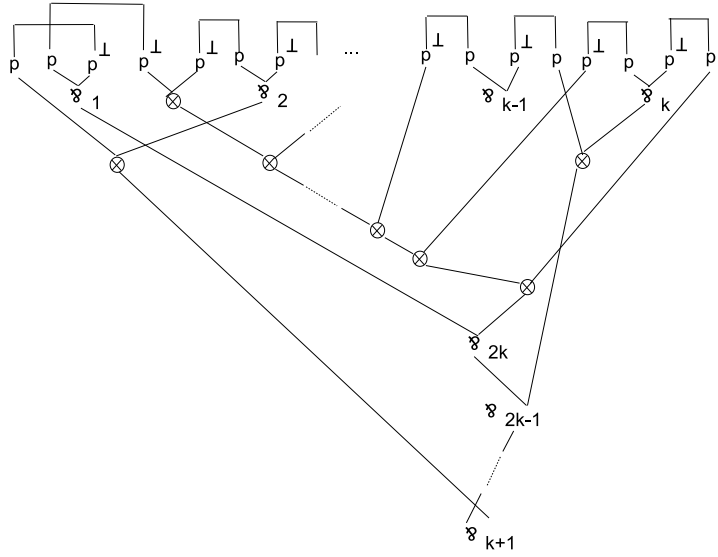}
\end{center}
 \caption{MLL Proof Net $\Theta_4$}
 \label{figPN2}
\end{figure}
  The proof net $\Theta$ is translated to the deNM-tree $T(\Theta)$ shown in Figure~\ref{figQuadraticEx2-0}.
  \begin{figure}[htbp]
\begin{center}
  \includegraphics[scale=0.5]{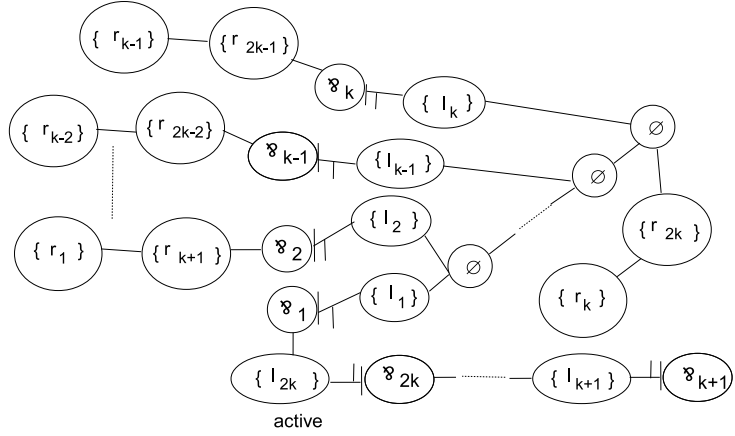}
\end{center}
 \caption{deNM-Tree $T(\Theta_4)$}
 \label{figQuadraticEx2-0}
\end{figure}
  The deNM-tree $T(\Theta_4)$ needs $O(k^2)$ computations in $\MCAL{R}$ when starting from the node labeled by $\{ \ell_{2k} \}$ at the worst scenario.
  Let us start rewriting based on our reduction strategy.
  In the starting active node has $\MCAL{Q}_{\rm down} = 1$ and $\MCAL{Q}_{\rm up} = 2k$.
  On the other hand, $\MCAL{Q}_{\rm labeled}, \MCAL{Q}_{\rm right}, \MCAL{S}_{\rm right}, \MCAL{Q}_{\rm up}$, and $\MCAL{S}_{\rm up}$ are
  all empty.
  After several rewriting steps, we reach the deNM tree shown in Figure~\ref{figQuadraticEx2-1}.
  Then for example, the active node has
  $\MCAL{Q}_{\rm up} = 1,2, \ldots, k-1, k$ and $\MCAL{Q}_{\rm right} = r_{2k}, r_{k}$.
  On the other hand, $\MCAL{Q}_{\rm down}, \MCAL{Q}_{\rm labeled}, \MCAL{S}_{\rm right}$, and $\MCAL{S}_{\rm up}$ are
  all empty.
  Our reduction strategy has some choices about which passive node is chosen for application of the union rule.
  But any choice leads to linear time termination.
  To each element in $\MCAL{Q}_{\rm up}$, application of the $\PAR$-elimination rule is tried.
  But $\MCAL{S}_{\rm right}$ is empty, all attempts fail.
  Then $\MCAL{S}_{\rm up}$ becomes $\{ 1,2, \ldots, k-1, k \}$.
  Next to each element in $\MCAL{Q}_{\rm right}$ application of the $\PAR$-elimination rule is tried.
  Then $r_{2k}$ fails, but $r_{k}$ succeeds.
  Then we get the deNM-tree $T_1$ shown in Figure~\ref{figQuadraticEx2-2}.
  In the new active node, $\MCAL{S}_{\rm up}$ is $\{ 1,2, \ldots, k-1, k \}$
  because we do not try to delete $k$.
  In addition, $\MCAL{S}_{\rm right} = \{ r_{2k} \}$ and $\MCAL{Q}_{\rm labeled}$ has one node labeled by $\{ r_{k+1}, r_{2k+1} \}$.
  The others are empty. 
  \begin{figure}[htbp]
    \begin{center}
      \includegraphics[scale=0.5]{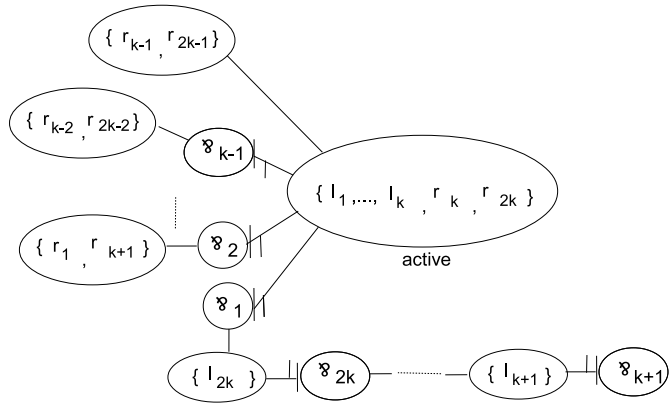}
    \end{center}
    \caption{deNM-Tree $T_1$}
    \label{figQuadraticEx2-2}
  \end{figure}

  \begin{figure}[htbp]
    \begin{center}
      \includegraphics[scale=0.5]{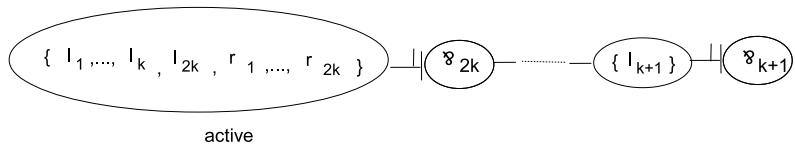}
    \end{center}
    \caption{deNM-Tree $T_2$}
    \label{figQuadraticEx2-3}
  \end{figure}
  After several steps, we reach the deNM-tree $T_2$ shown in Figure~\ref{figQuadraticEx2-3}.
  Up to now, we have tried application of the $\PAR$-elimination $3 k$ times.
  In the active node, we have
  \begin{eqnarray*}
    \MCAL{S}_{\rm right} & = & \{ r_{k+1}, \ldots, r_{2k} \} \\
    \MCAL{Q_{\rm up}} & = & 2k  \\
    \MCAL{S}_{\rm up} & = & \{ 1,2, \ldots, k-1, k \}.
  \end{eqnarray*}
  The others are empty.
  The for the only element $2k$ in $\MCAL{Q_{\rm up}}$ application of the $\PAR$-elimination rule is tried
  and succeeds.
  After several steps, we finally obtain one node tree.
  In total from the start to the end, we have tried application of the $\PAR$-elimination $4 k$ times, which is linear.
\end{example}
\begin{example}
  As seen previously, the deNM-tree $T(\Theta_3)$ shown in Figure~\ref{fignonPN-loop-1-Tree}
  is obtained from MLL proof structure $\Theta_3$ shown in Figure~\ref{fignonPN-loop-1}, which is not an MLL proof net.
  When we choose the node $n$ labeled by $\{ r_2 \}$ as the starting rule,
  since $\MCAL{Q}_{\rm down} = 1$ for $n$, i.e., it is not empty,
  the local jump rule is applied to $\PAR$-link $1$.
  After two applications of the union rule, in the current active node, $\MCAL{Q}_{\rm down}$, $\MCAL{Q}_{\rm labeled}$, $\MCAL{Q}_{\rm right}$, and $\MCAL{Q}_{\rm up}$ become all empty.
  So the output is {\bf no}. 
\end{example}

\section{Concluding Remarks}
In this paper we have accomplished a new linear-time correctness condition of unit-free MLL proof nets based on the rewriting system over trees.
Among known linear-time correctness conditions of MLL, ours is definitely simplest.
We already given a prototype implementation for our algorithm in \cite{Mat19a}.
Compared with them based on Girard's original sequentialization definition and original de Naurois and Mogbil's correctness condition,
our new implementation is remarkably faster:
MLL Proof nets that cannot be checked by them in a week, can be checked by ours in a minute!

Although it is not trivial whether
a linear time sequentialization algorithm is derived from our linear time correctness condition,
we have already obtained such an algorithm.
The topic will be given elsewhere. 

There are some future research directions. 
\begin{itemize}
\item Extensions of our result to variants like noncommutative fragments or extensions like MALL or MELL.
\item Implementation issues:
  In particular, to some extent it may be possible to have several active nodes in a deNM-tree and to exploit parallelism using one or many multi-core processors.
\item Application to proof search:
  In \cite{Mat19a} in order to search MLL proof nets for a given MLL formula,
  a backtracking mechanism and a naive implementation of de Naurois and Mogbil's correctness condition are combined.
  Our result may be used to obtain more elegant implementations for MLL proof search.
  That was our original motivation for this work. 
\item Mechanical formalization using your favorite interactive theorem prover.
\end{itemize}

\bibliographystyle{amsplain}
\bibliography{generic}

\end{document}